\documentclass[12pt]{article}

\usepackage[a4paper, includefoot,
            left=3cm, right=3cm,
            top=2cm, bottom=2cm,
            headsep=1cm, footskip=1cm]{geometry}

\usepackage{graphics}
\usepackage{subcaption}
\usepackage{tikz}
\usepackage{url}
\RequirePackage[colorlinks]{hyperref}

\usepackage{amsmath}
\usepackage{latexsym,amssymb, amsthm}

\usepackage{euscript}

\newcommand{\cM}{\EuScript{M}}

\newcommand{\rmF}{\mathrm{F}}

\newcommand{\rmT}{\mathrm{T}}

\newcommand{\tR}{\mathbb{R}}
\newcommand{\tS}{\mathbb{S}}

\newcommand{\tX}{\mathbb{X}}
\newcommand{\tY}{\mathbb{Y}}
\newcommand{\tZ}{\mathbb{Z}}

\newcommand{\bfA}{\mathbf{A}}

\newcommand{\bfC}{\mathbf{C}}
\newcommand{\bfD}{\mathbf{D}}
\newcommand{\bfE}{\mathbf{E}}
\newcommand{\bfF}{\mathbf{F}}

\newcommand{\bfL}{\mathbf{L}}

\newcommand{\bfO}{\mathbf{O}}
\newcommand{\bfP}{\mathbf{P}}
\newcommand{\bfQ}{\mathbf{Q}}
\newcommand{\bfR}{\mathbf{R}}

\newcommand{\bfU}{\mathbf{U}}
\newcommand{\bfV}{\mathbf{V}}

\newcommand{\bfX}{\mathbf{X}}
\newcommand{\bfY}{\mathbf{Y}}
\newcommand{\bfZ}{\mathbf{Z}}

\newcommand{\calH}{\mathcal{H}}

\newcommand{\calT}{\mathcal{T}}

\newcommand{\bt}{\begin{theorem}}
\newcommand{\et}{\end{theorem}}
\newcommand{\bl}{\begin{lemma}}
\newcommand{\el}{\end{lemma}}
\newcommand{\bp}{\begin{proposition}}
\newcommand{\ep}{\end{proposition}}
\newcommand{\bc}{\begin{corollary}}
\newcommand{\ec}{\end{corollary}}

\newcommand{\bd}{\begin{definition}\rm}
\newcommand{\ed}{\end{definition}}
\newcommand{\bex}{\begin{example}\rm}
\newcommand{\eex}{\end{example}}
\newcommand{\br}{\begin{remark}\rm}
\newcommand{\er}{\end{remark}}

\newcommand{\btbh}{\begin{table}[!ht]}
\newcommand{\etb}{\end{table}}
\newcommand{\bfgh}{\begin{figure}[!ht]}
\newcommand{\efg}{\end{figure}}

\newcommand{\bea}{\begin{eqnarray*}}
\newcommand{\eea}{\end{eqnarray*}}
\newcommand{\be}{\begin{eqnarray}}
\newcommand{\ee}{\end{eqnarray}}

\newcommand{\suml}{\sum\limits}

\newcommand{\ve}{\varepsilon}

\newcommand{\lm}{\lambda}
\def\wtilde{\widetilde}
\def\what{\widehat}

\newcommand{\ra}{\rightarrow}

\def\spaceR{\mathsf{R}}
\def\spaceC{\mathsf{C}} 

\newcommand{\bfw}{\mathbf{w}}

\def\rank{\mathop{\mathrm{rank}}}

\def\tr{\mathop{\mathrm{tr}}}

\newcommand{\diag}{\mathop{\mathrm{diag}}}

\makeatletter
\def\adots{\mathinner{\mkern2mu\raise\p@\hbox{.}
\mkern2mu\raise4\p@\hbox{.}\mkern1mu
\raise7\p@\vbox{\kern7\p@\hbox{.}}\mkern1mu}}
\newcommand{\l@abcd}[2]{\hbox to\textwidth{#1\dotfill #2}}
\makeatother

\newtheorem{proposition}{Proposition}
\newtheorem{corollary}{Corollary}
\newtheorem{theorem}{Theorem}
\newtheorem{remark}{Remark}
\newtheorem{lemma}{Lemma}
\newtheorem{definition}{Definition}
\newtheorem{algorithm}{Algorithm}

\begin{document}

\title{Variations of Singular Spectrum Analysis for separability improvement: non-orthogonal decompositions of time series}

\author{Nina Golyandina \footnote{Corresponding author, nina@gistatgroup.com}, Alex Shlemov \footnote{shlemovalex@gmail.com}}

\date{\small Department of Statistical Modelling,
    Department of Statistical Modeling, St.~Petersburg State University,
    Universitetsky pr 28, St.~Petersburg 198504, Russia}

\maketitle

\begin{abstract}
Singular spectrum analysis (SSA) as a nonparametric tool for decomposition of an observed time series into sum of
interpretable components such as trend, oscillations and noise is considered.
The separability of these series components by SSA means the possibility of such decomposition.
Two variations of SSA, which weaken the separability conditions, are proposed. Both proposed approaches
consider inner products corresponding to oblique coordinate systems instead of the conventional
Euclidean inner product. One of the approaches performs iterations to obtain separating inner products.
The other method changes contributions of the components by involving the series derivative
to avoid component mixing. Performance of the suggested methods is demonstrated on simulated and real-life data.
\end{abstract}


\section{Introduction}

Singular spectrum analysis
\cite{Broomhead.King1986, Elsner.Tsonis1996, Ghil.etal2002, Golyandina.etal2001, Golyandina.Zhigljavsky2012, Vautard.etal1992, Zhigljavsky2010}
is a powerful method of time series analysis, which does not
require a parametric model of the time series given in advance and therefore SSA is very well suitable for exploratory analysis.
After an exploratory analysis has been performed, SSA enables to construct series models.

Singular spectrum analysis can solve very different problems in time series analysis which range from
the series decomposition on interpretable series components to forecasting,
missing data imputation, parameter estimation and many others.
The main problem is the proper decomposition of the time series. For example, if one forecasts trend, then
this trend should be extracted properly. For seasonal adjustment, the seasonality should be extracted correctly, and so on.

In \cite{Golyandina.etal2001, Nekrutkin1997}, the separability theory, which is responsible
for the proper decomposition and proper component extraction, was developed.
The separability of components means that the method is able to extract the time series components
from the observed series that is the sum of many components.
At the present time, there is a lot of publications with theory of separability and applications where separability is important, see
\cite{AtikurRahmanKhan.Poskitt2013, Awichi.Mueller2013, Gillard.Knight2013, Golyandina2010, Hassani.Zhigljavsky2009, Hudson.Keatley2010,
Itoh.Kurths2011, Itoh.Marwan2013, Nekrutkin2010, Patterson.etal2011, Ruch.Bester2013}
among others.

For reasonable time series lengths and noise levels,
trends, oscillations and noise are approximately separable by SSA \cite[Sections 1.5 and 6.1]{Golyandina.etal2001}.
However, the conditions of approximate separability can be restrictive, especially, for
short time series.

The separability conditions are closely related to the properties of the singular value decomposition (SVD), which is the essential part
of many statistical and signal processing methods: principal component analysis \cite{Jolliffe2002},
low-rank approximations \cite{Markovsky2012}, several subspace-based methods \cite{Veen.etal1993} including
singular spectrum analysis among many others.
The main advantage of the SVD is its optimality features and bi-orthogonality; the drawback for
approximation problems is the non-uniqueness of the SVD expansion if there  are coinciding singular values.

In subspace-based methods, the SVD is applied to a trajectory matrix with rows and columns consisting of
subseries of the initial series. In SSA, the obtained SVD components are grouped and the grouped matrices are
transferred back to the series. Thus, we obtain a decomposition of the initial time series $\tX$
into a sum of series components, e.g., $\tX=\widetilde\tS+\widetilde\tR$.
If we deal with a series $\tX=\tS+\tR$ containing two series components $\tS$ and $\tR$, which we want to find,
then (approximate) weak separability is by definition (approximate) orthogonality of subseries of $\tS$ and $\tR$,
which provides, due to the SVD bi-orthogonality, the existence of such a grouping that $\widetilde\tS$ and $\widetilde\tR$
are (approximately) equal to  $\tS$ and $\tR$ correspondingly.

Non-uniqueness  of the SVD in the case of coinciding singular values
implies the condition of disjoint sets of singular values in the groups corresponding to
different series components to avoid their possible mixture. This condition is necessary to obtain the so called strong
separability, when any SVD of the trajectory matrix provides the proper grouping.
In practice, the strong separability is needed
(see for more details Section~\ref{sec:sep}) and both conditions, orthogonality of component subseries and
disjoint sets of singular values of component trajectory matrices, should be fulfilled.

The paper presents two methods, Iterative O-SSA and DerivSSA, which help to weaken the separability conditions in SSA.
For simplicity, we describe the methods for separation of two series components;
separation of several components can be considered in analogous manner.

Orthogonality of subseries can be a strong limitation on the separated series. However, if we consider
orthogonality with respect to non-standard Euclidean inner product, conditions of separability
can be considerably weakened. This yields the first method called Oblique SSA (O-SSA) with the SVD step performed in a non-orthogonal
coordinate system.
The idea of Iterative Oblique SSA is similar to prewhitening that is frequently used in statistics as preprocessing:
if we know covariances between components, then we can perform linear transformation and obtain
uncorrelated components.
Since the `covariances' of the components are not known in advance, the iterative algorithm called Iterative Oblique SSA is suggested.
Contribution of the components can be changed in a specific way during the iterations to improve separability.

The second method called DerivSSA helps to change the component contributions with no change of the structure of the separated
series components. The approach consists in consideration of the series derivative together with the series itself.
For example, two singular values produced by a sinusoid are determined by its amplitude.
The derivative of a sine wave has the same frequency and changed amplitude, depending
on frequency: $f(x)=\sin(2\pi\omega x + \phi)$ has amplitude $1$, while its derivative has
amplitude $2\pi\omega$. This is just a simple example; the method works with non-stationary series,
not only with sinusoids.
The use of derivatives helps to overcome the problem when the approximate orthogonality holds but the series components
mix due to equal contributions.
It seems that this approach is simpler and  more general than the SSA-ICA (SSA with Independent Component Analysis) approach considered in \cite[Section 2.5.4]{Golyandina.Zhigljavsky2012}.

Since both suggested approaches do not have approximating features, they cannot replace Basic SSA
and therefore should be used in a nested manner. This means that Basic SSA extracts mixing series
components (e.g. first we use Basic SSA for denoising) and then one of the proposed methods separates the
mixing components. Let us demonstrate the nested use of the methods by an example. Let $\tX=(x_1,\ldots,x_{N})$ be the series of length $N$,
$\tX={\tX}^{(1)}+{\tX}^{(2)}+{\tX}^{(3)}$. The result of Basic SSA is
$\tX=\widetilde{\tX}^{(1,2)}+\widetilde{\tX}^{(3)}$, the result of the considered method is
$\widetilde{\tX}^{(1,2)}=\widetilde{\tX}^{(1)}+\widetilde{\tX}^{(2)}$ and the final result is
$\tX=\widetilde{\tX}^{(1)}+\widetilde{\tX}^{(2)}+\widetilde{\tX}^{(3)}$.

The paper is organized as follows. We start with a short description of the algorithm of Basic SSA and standard separability notion (Section~\ref{sec:BasicSSA}).
The next two sections~\ref{sec:ObliqueSSA} and \ref{sec:DerivSSA} are devoted to the variations of singular spectrum analysis.
In Section~\ref{sec:ObliqueSSA}, Oblique SSA is considered. In Section~\ref{sec:DerivSSA},
SSA involving series derivatives is investigated.
Each section contains numerical examples of algorithm application.
In Section~\ref{sec:examples}, both methods are applied to real-life time series.
Conclusions are contained in Section~\ref{sec:conclusion}.
Since the methods are based on the use of inner products and decompositions in oblique coordinate systems, we put the necessary definitions and statements into Appendix~\ref{sec:app}.

An implementation of the proposed algorithms is contained in the  \textsc{R}-package \textsc{Rssa} as of version 0.11 \cite{Korobeynikov.etal2014},
which is thoroughly described for Basic SSA in \cite{Golyandina.Korobeynikov2013}.
Efficiency of the implementation of Basic SSA and its variations is based on the use of the approach described in  \cite{Korobeynikov2010}.
The code for most of the presented examples can be found in the documentation of \textsc{Rssa}.

\section{Basic SSA}
\label{sec:BasicSSA}
\subsection{Algorithm}

Consider a real-valued time series $\tX=\tX_N=(x_1,\ldots,x_{N})$ of length $N$.
Let $L$ ($1<L<N$) be some integer called {\em window length} and $K=N-L+1$.

For convenience, denote $\cM_{L,K}$ the space of matrices of size $L\times K$,
$\cM_{L,K}^{(H)}$ the space of Hankel matrices of size $L\times K$,
$X_i=(x_{i},\ldots,x_{i+L-1})^\rmT$, $i=1,\ldots,K$, the {\em $L$-lagged vectors}
and $\bfX=[X_1:\ldots:X_K]$ the {\em $L$-trajectory matrix} of the series $\tX_N$.
Define the embedding operator $\calT: \spaceR^{N} \mapsto \cM_{L,K}$
as $\calT(\tX_N)=\bfX$.

Also introduce the projector $\calH$ (in the Frobenius norm) of
 $\cM_{L,K}$ to $\cM_{L,K}^{(H)}$, which performs the projection by the change of
 entries on auxiliary
diagonals $i+j=\mathrm{const}$ to their averages along the diagonals.

The algorithm of Basic SSA consists of four steps.

\smallskip
{\bf 1st step: Embedding}.
Choose $L$. Construct the $L$-trajectory matrix:
 $\bfX=\calT(\tX_N)$.

\smallskip
{\bf 2nd step: Singular value decomposition (SVD)}.
Consider the SVD of the trajectory matrix:
\be
\label{eq:elem_matr}
\bfX=\sum_{i=1}^d \sqrt{\lambda_i}U_i V_i^\rmT=\bfX_1 + \ldots + \bfX_d,
\ee
where $\sqrt{\lambda_i}$ are singular values,
$U_i$ and $V_i$ are the left and right singular vectors of $\bfX$,
$\lambda_1\geq\ldots\geq \lambda_d > 0$, $d=\rank(\bfX)$. The number $d$ is called
\emph{$L$-rank} of the series $\tX$.

The triple $(\sqrt{\lm_i},U_i,V_i)$ is called $i$th {\em
eigentriple} (abbreviated as ET).

\smallskip
{\bf 3rd step: Eigentriple grouping}.
The grouping procedure
partitions the set of indices $\{1,\ldots,d\}$ into $m$ disjoint subsets
$I_1,\ldots,I_p$. This step is less formal. However, there are different
recommendations on grouping related to separability issues
briefly described in Section~\ref{sec:sep}.

Define $\bfX_I=\sum_{i\in I} \bfX_i$.
The expansion \eqref{eq:elem_matr} leads to the decomposition
\be
\label{eq:mexp_g}
\bfX=\bfX_{I_1}+\ldots+\bfX_{I_p}.
\ee

If $p=d$ and $I_j=\{j\}$,
$j=1,\ldots,d$, then the corresponding grouping is called \emph{elementary}.

\smallskip
{\bf 4th step: Diagonal averaging}.
Obtain the series by diagonal averaging of the matrix components of \eqref{eq:mexp_g}:
$\wtilde\tX^{(k)}_N = \calT^{-1} \calH \bfX_{I_k}$.

\smallskip
Thus, the algorithm yields the decomposition of the observed time series
\be
\label{eq:sexp_f}
  \tX_N = \suml_{k=1}^p\wtilde\tX^{(k)}_N.
\ee
The reconstructed components produced by the elementary grouping will be called
\emph{elementary reconstructed series}.

\subsection{Separability by Basic SSA}
\label{sec:sep}
Notion of separability is very important to understand how SSA works.
Separability of two time series $\tX^{(1)}_N$ and $\tX^{(2)}_N$ signifies the possibility of extracting
$\tX^{(1)}_N$ from the observed series $\tX_N=\tX^{(1)}_N + \tX^{(2)}_N$.
This means that there exists a grouping at Grouping step such that
$\wtilde\tX^{(m)}_N=\tX^{(m)}_N$.

Let us define the separability formally.  Let $\bfX^{(m)}$ be the trajectory matrices of the considered series,
$\bfX^{(m)}=\sum_{i=1}^{d_m} \sqrt{\lambda_{m,i}}U_{m,i} V_{m,i}^\rmT$, $m=1,2$, be their SVDs.
The column and row spaces of the trajectory matrices are called \emph{column} and \emph{row trajectory spaces} correspondingly.

\begin{definition}
Let $L$ be fixed. Two series $\tX^{(1)}_N$ and $\tX^{(2)}_N$ are called \emph{weakly separable}, if their column trajectory spaces
are orthogonal and the same is valid for their row trajectory spaces, that is, $(\bfX^{(1)})^\rmT \bfX^{(2)}=\bf0_{K,K}$ and
 $\bfX^{(1)} (\bfX^{(2)})^\rmT =\bf0_{L,L}$.
\end{definition}

\begin{definition}
Two series $\tX^{(1)}_N$ and $\tX^{(2)}_N$ are called \emph{strongly separable}, if they are weakly separable
and the sets of singular values of their $L$-trajectory matrices are disjoint, that is, $\lambda_{1,i} \neq \lambda_{2,j}$ for any $i$ and $j$.
\end{definition}

By definition, separability means orthogonality of the column and row spaces
of the trajectory matrices of the series components $\tX^{(1)}_N$ and $\tX^{(2)}_N$.
For approximate (asymptotic) separability with $\wtilde\tX^{(m)}_N\approx \tX^{(m)}_N$
we need the condition of approximate (asymptotic) orthogonality of subseries of the considered components.
Asymptotic separability is considered as $L,K,N \ra \infty$.

For sufficiently long time series, SSA can approximately separate, for example, signal and noise, sine waves with different
frequencies, trend and seasonality \cite{Golyandina.etal2001,
    Golyandina.Zhigljavsky2012}.

Let us demonstrate the separability of two sinusoids with frequencies $\omega_1$
and $\omega_2$: $x^{(i)}_n=A_i\cos(2\pi\omega_i n+\phi_i)$. These sinusoids are asymptotically separable, that is, their subseries
are asymptotically orthogonal as their length tends to infinity.
However, the rate of convergence depends on the difference between the frequencies.
If they are close and the time series length is not long enough, the series can be far from
orthogonal and therefore not separable.

Weak separability means that at SVD step there exists
such an SVD that admits the proper grouping. The problem of possibility of a non-separating SVD expansion is related to non-uniqueness of the SVD
in the case of equal singular values.
Strong separability means that any SVD of the series trajectory matrix admits the proper grouping.
Therefore, we need strong (approximate) separability for the use in practice.
For example, two sinusoids with equal amplitudes are asymptotically weakly separated,
but asymptotically not strongly separated and therefore are mixed in the decomposition.

\subsubsection{Separability measure}
\label{sec:sep_mes}
Very helpful information for  detection of separability and group identification is contained in the so-called
$\bfw$-correlation matrix.
This matrix consists of weighted
cosines of angles between the reconstructed time series components. The weights
reflect the number of entries of the time series terms into its trajectory
matrix.

Let $w_n=\#\{(i,j): 1\leq i \leq L, 1\leq j \leq K, i+j=n+1\}$. Define the $\bfw$-scalar product of time series of length $N$
as $(\tY_N, \tZ_N)_\bfw=\sum_{n=1}^N w_n y_n z_n = \langle \bfY, \bfZ\rangle_\rmF$. Then
\bea
\rho_\bfw (\tY_N, \tZ_N) = (\tY_N, \tZ_N)_\bfw/(\|\tY_N\|_\bfw \|\tZ_N\|_\bfw).
\eea

Well separated components in \eqref{eq:sexp_f} have small correlation whereas poorly separated
components generally have large correlation. Therefore, looking at the
matrix  of $\bfw$-correlations between elementary reconstructed series $\wtilde\tX^{(i)}_N$ and $\wtilde\tX^{(j)}_N$ one can find groups of correlated
series components and use
this information for the consequent grouping. One of the rules is not to include
the correlated components into different groups.
Also, $\bfw$-correlations can be used for checking the grouped decomposition.

It is convenient to depict in absolute magnitude the matrix of $\bfw$-correlations between the series components
graphically in the white-black scale, where small correlations are shown in white, while correlations with moduli close to 1
are shown in black.

\subsubsection{Scheme of Basic SSA application}

Let us briefly describe the general scheme of Basic SSA application, thoroughly described in
\cite{Golyandina.etal2001,Golyandina.Zhigljavsky2012}:

\begin{itemize}
\item
Choice of window length $L$ in accordance with a-priori recommendations (see, in addition, \cite{Golyandina2010}).

\item
Execution of Embedding and Decomposition steps.

\item
Analysis of the eigentriples and the $\bfw$-correlation matrix to perform grouping of eigentriples.
The main principle is: eigenvectors repeat the form of a series component that produces these eigentriples. $\bfw$-Correlations
also provide a guess for proper grouping.

\item
Execution of Grouping and Reconstruction steps to obtain the desired series decomposition.

\item
If separability does not take place for the given $L$ and the obtained decomposition is not appropriate, then the change of the window length $L$
is recommended.
\end{itemize}

Note that the proper grouping to obtain a suitable series decomposition can be impossible if the signal components (described, as a rule,
by a number of leading SVD components) are mixed. For example,
if a signal eigenvector contains both periodic and slowly varying components, this means
that the trend and periodic components are not separable, at least for the chosen window length $L$.
If we see the mixture of two sine-waves with different frequencies, this means that
these sine-waves are not separable for this $L$.

If it appears that for the chosen $L$ there is no separability (weak or strong),
the attempt to obtain separability  is performed with other choices of $L$.
For example, a possible lack of strong separability between a trend of complex form and a seasonality
can be overcome by means of the use of small window lengths. However, weak separability can be weakened by this trick and
Sequential SSA should be used to obtain an accurate decomposition of the residual after the trend extraction.

For the majority of time series, SSA with a proper choice of window length is able to separate series components
and to obtain a desirable series decomposition. However, sometimes Basic SSA cannot separate certain components
such as short sine wave series with close frequencies or sine waves with equal amplitudes.

\subsubsection{Identification of separated sinusoids}

Separation of sine-wave components is of special interest. Each sine-wave component generates two elementary series components,
which have correlation close to 1.
If a sinusoid is separated from the residual, maybe, approximately, then two elementary components produced by it
are almost not correlated with the other elementary components
and therefore we will see a black square $2\times 2$ on the $\bfw$-correlation matrix
of elementary components.

To find two SVD components corresponding to a sine-wave, scatterplots of
eigenvectors (which are approximately sine and cosine) can be also used.
If the period value is integer, the scatterplot of sine vs cosine looks like a regular polygon, where the number of vertices
is equal to the period.

\begin{figure}[!htb]
    \centering
    \begin{subfigure}{.32\columnwidth}
        \centering
        \includegraphics{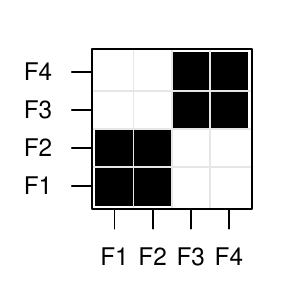}
        \caption*{(A)}
    \end{subfigure}
    \begin{subfigure}{.32\columnwidth}
        \centering
        \includegraphics{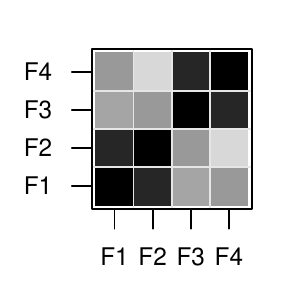}
        \caption*{(B)}
    \end{subfigure}
    \begin{subfigure}{.32\columnwidth}
        \centering
        \includegraphics{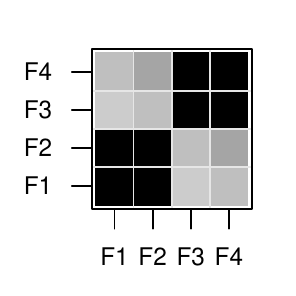}
        \caption*{(C)}
    \end{subfigure}
    \caption{Sum of two sinusoids: $\bfw$-correlation matrices for different types of separability}
    \label{fig:demosin_wcor}
\end{figure}

\begin{figure}[!htb]
    \centering
    \begin{subfigure}{0.48\linewidth}
        \centering
        \includegraphics{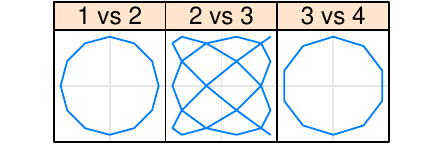}
        \caption*{(A)}
    \end{subfigure}
    \begin{subfigure}{0.48\linewidth}
        \centering
        \includegraphics{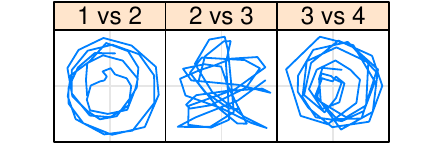}
        \caption*{(C)}
    \end{subfigure}
    \caption{Sum of two sinusoids: scatterplots of eigenvectors with good (left) and bad (right) separability}
    \label{fig:demosin_scatter}
\end{figure}

For example, consider the series $\tX_N$,
where $x_n=x_n^{(1)}+x_n^{(2)}$, $x_n^{(k)}=A_k \sin(2\pi \omega_k n)$, the series length $N=119$, with three different sets of parameters: \\
(A) `strong separability', $A_1=2$, $A_2=1$, $\omega_1=1/12$, $\omega_2=1/10$; \\
(B) 'weak separability, no strong separability', $A_1=A_2=1$, $\omega_1=1/12$, $\omega_2=1/10$;\\
(C) 'no weak separability', $A_1=2$, $A_2=1$, $\omega_1=1/12$, $\omega_1=1/13$, the series is corrupted by Gaussian white noise
with standard deviation 4.

 The difference between good and bad separability is clearly seen in Fig.~\ref{fig:demosin_wcor} and \ref{fig:demosin_scatter}.
One can see that the matrices of weighted correlations for the examples (B) and (C) are very similar, although
in general weighted correlations for the example (B) can be arbitrary.
Figure~\ref{fig:demosin_scatter} shows the scatterplots of eigenvectors for the examples (A) and (C).
The pairs of eigenvectors produced by exactly separated sinusoids form regular polygons.

\subsection{Series of finite rank and series governed by linear recurrence relations}
\label{sec:lrr}
Let us describe the class of series of finite rank, which is natural for SSA.
Note that only such time series can be exactly separated and exactly continued by SSA
\cite[Section 2.2 and Chapter 5]{Golyandina.etal2001}.

We define $L$-rank of a series $\tX_N$ as the rank of its $L$-trajectory matrix.
Series with rank-deficient trajectory matrices are of special interest.
A time series is called \emph{time series of finite rank} $r$ if its
$L$-trajectory matrix has rank $r$ for any $L\ge r$ (it is convenient to
assume that $L\le K$).

Under some not restrictive conditions, a series $\tS_N$ of finite rank $r$
is governed by a linear recurrence relation (LRR) of order $r$, that is
\be
\label{eq:lrf}
s_{i+r}=\sum_{k=1}^r a_k s_{i+r-k},\ 1\leq i\leq N-r,\ a_r\neq 0.
\ee

The LRR~\eqref{eq:lrf} is called minimal, since it is unique and has minimal order among LRRs governing $\tS_N$.
Let us describe how we can restore the form of the time series by means of the minimal LRR.

\begin{definition}
    A polynomial $P_r(\mu)=\mu^r - \sum_{k=1}^r a_k \mu^{r-k}$ is called a
    {characteristic polynomial} of the LRR \eqref{eq:lrf}.
\end{definition}
Let the time series $\tS_{\infty}=(s_1,\ldots,s_n,\ldots)$ satisfy the LRR
\eqref{eq:lrf} for $i\geq 1$. Consider the characteristic
polynomial of the LRR \eqref{eq:lrf} and denote its different (complex) roots by
$\mu_1,\ldots,\mu_p$, where $p \leq r$.  All these roots are non-zero, since
$a_r\neq 0$.  Let the multiplicity of the root $\mu_m$ be $k_m$, where $1\leq
m\leq p$ and $k_1+\ldots+k_p=r$.
We will call the set $\{\mu_j\}_{m=1}^p$ \emph{characteristic (or signal) roots} of the series governed by an LRR.
Note that in the framework of SSA non-minimal LRRs, which have so called extraneous roots in addition to the signal ones, are considered
and the extraneous roots are studied (\cite{Usevich2010}); however, here we will deal only with characteristic roots to describe
the signal model.

It is well-known that the time series $\tS_{\infty}=(s_1,\ldots,s_n,\ldots)$
satisfies the LRR $\eqref{eq:lrf}$ for all $i\ge 0$ if and only if
\be
\label{eq:GEN_REQ}
s_n = \suml_{m=1}^p \left(\suml_{j=0}^{k_m-1} c_{m,j} n^j\right) \mu_m^n.
\ee
for some $c_{m,j}\in \spaceC$. For real-valued time series, \eqref{eq:GEN_REQ} implies
that the class of time series governed by LRRs consists of sums of products
of polynomials, exponentials and sinusoids.

The important advantage of SSA is that although the model \eqref{eq:GEN_REQ} of signals is involved in theoretical results,
the SSA algorithm does not perform explicit estimation of the model parameters for reconstruction and forecasting.
This provides the possibility to deal with signals that are locally approximated by the model; in particular,
to extract slowly-varying trends and modulated sine waves.
The indicated feature of the SSA approach holds for the variations considered below.

\section{Oblique SSA}
\label{sec:ObliqueSSA}
Although many interpretable series components like trend (a slowly varying component) and seasonality are asymptotically orthogonal,
for the given time series length the orthogonality can be not reached even approximately.
Therefore, it would be helpful to weaken the orthogonality condition.
The suggested approach consists in using an orthogonality, which still means the equality of an inner product to 0,
but this is a non-ordinary inner product which is adapted to time series components, which we want
to separate.

It is well-known that any inner product in Euclidean space is associated with a symmetric positive-definite matrix
$\bfA$ and is defined as $\langle X_1, X_2 \rangle_\bfA=(\bfA X_1,X_2)$.
The standard inner product is given by the identity matrix.
Inner product implies $\bfA$-ortho\-go\-na\-li\-ty of the vectors if $\langle X_1, X_2 \rangle_\bfA = 0$.
If the matrix $\bfA$ is semi-definite, then it produces the inner product given in its column (or row, it is the same due to symmetry) space.
Below, considering $\langle X_1, X_2 \rangle_\bfA$, we will always assume that the vectors $X_i$, $i=1,2$, belong to the column space of $\bfA$.

Thus, non-standard Euclidean inner products induce such notions as oblique coordinate systems,
ortho\-go\-na\-li\-ty of vectors, which are oblique in ordinary sense, and so on.

Let us present an elementary example. Let $X=(1,2)^\rmT$ and $Y=(1,1)^\rmT$. Certainly, these vectors are
not orthogonal in the usual sense: $(X,Y)=3$. However, if we define
\begin{equation}
\label{eq:A}
\bfA=
\left(\begin{array}{rr}
5 &-3\\
-3 &2
\end{array}\right),
\end{equation}
then $\langle X, Y\rangle_\bfA =(\bfA X, Y)=0$ and $(\bfO_\bfA X,\bfO_\bfA Y)=0$ for any $\bfO_\bfA$ such that $\bfO_\bfA^\rmT \bfO_\bfA=\bfA$, e.g.
\begin{equation*}
\bfO_\bfA=
\left(\begin{array}{rr}
-1 &1\\
2 & -1
\end{array}\right).
\end{equation*}
This means that $\{X,Y\}$ is an orthogonal basis with respect to the $\bfA$-inner product $\langle \cdot, \cdot \rangle_\bfA$
and $\bfO_\bfA$ corresponds to an orthogonalizing map. The matrix $\bfA$ can be chosen such that
$X$ and $Y$ have any $\bfA$-norm. The choice \eqref{eq:A} corresponds to $\bfA$-orthonormality.

\smallskip
To describe a so called Oblique SSA, let us introduce the SVD of a matrix $\bfX$ produced by two oblique bases,
$\bfL$-orthonormal and $\bfR$-orthonormal correspondingly, in the row and column spaces (Definition~\ref{def:LRSVD}).
We say that $\bfX=\sum_{i=1}^d \sigma_i P_i Q_i^\rmT$
is the $(\bfL,\bfR)$-SVD, if $\{P_i\}_{i=1}^d$ is an $\bfL$-orthonormal system and $\{Q_i\}_{i=1}^d$ is an $\bfR$-orthonormal system, that is, the decomposition
is $(\bfL,\bfR)$-biorthogonal. This kind of SVD is called Restricted SVD (RSVD) given by the triple $(\bfX,\bfL,\bfR)$, see \cite{DeMoor.Golub1991} for details.
Mathematics related to inner products $\langle \cdot, \cdot \rangle_\bfA$ with positive-semidefinite matrix $\bfA$
 and the corresponding RSVD is shortly described in Appendix~\ref{sec:app} from the viewpoint of decompositions into
 a sum of elementary matrices. We formulate
 the necessary definitions and propositions in a convenient form to make the suggested algorithms clearer.

\emph{Oblique SSA (O-SSA)} is the modification of the Basic SSA algorithm described in Section~\ref{sec:BasicSSA},
where the SVD step is changed by the $(\bfL,\bfR)$-SVD
for some matrices $\bfL$ and $\bfR$ consistent with $\bfX$ (see Definition~\ref{def:agreed}). We will use the notions introduced in the algorithm
of Basic SSA also for its oblique modification.

Proposition~\ref{prop:calcLRSVD} provides the algorithm which reduces the $(\bfL,\bfR)$-SVD to the ordinary SVD.

\begin{algorithm}
\label{alg:LRSVD}
\rm
\textbf{($(\bfL,\bfR)$-SVD.)}

\smallskip\noindent
\textbf{Input}: $\bfY$, $(\bfL,\bfR)$ consistent with $\bfY$.

\smallskip\noindent
\textbf{Output}: The $(\bfL,\bfR)$-SVD in the form \eqref{eq:LRSVD}.

\begin{enumerate}
\item
Calculate $\bfO_\bfL$ and $\bfO_\bfR$, e.g., by Cholesky decomposition.
\item
Calculate $\bfO_\bfL\bfY\bfO_\bfR^\rmT$.
\item
Find the ordinary SVD decomposition \eqref{eq:LVSVDcalc}.
\item
$\sigma_i=\sqrt{\lambda_i}$, $P_i=\bfO_\bfL^\dag U_i$ and $Q_i=\bfO_\bfR^\dag V_i$. where $\dag$ denotes pseudo-inverse.
\end{enumerate}

\end{algorithm}

Note that if $\bfL$ and $\bfR$ are the identity matrices, then Oblique SSA coincides with Basic SSA,
$\sigma_i=\sqrt{\lambda_i}$, $P_i=U_i$ and $Q_i=V_i$.

\subsection{Separability}
\label{sec:LRsepar}
The notion of weak and strong $(\bfL,\bfR)$-separability, which is similar to conventional separability
described in Section~\ref{sec:sep}, can be introduced.
Again, let $\tX=\tX^{(1)} + \tX^{(2)}$, $\bfX$ be its trajectory matrix, $\bfX^{(m)}$ be the trajectory matrices of the series components,
$\bfX^{(m)}=\sum_{i=1}^{r_m} \sigma_{m,i}P_{m,i} Q_{m,i}^\rmT$ be their $(\bfL,\bfR)$-SVDs, $m=1,2$.
We assume that  $\bfL$ and $\bfR$ are consistent with $\bfX$, $\bfX^{(1)}$ and $\bfX^{(2)}$.

\begin{definition}
Let $L$ be fixed. Two series $\tX^{(1)}_N$ and $\tX^{(2)}_N$ are called \emph{weakly $(\bfL,\bfR)$-separable}, if their column trajectory spaces
are $\bfL$-orthogonal and their row trajectory spaces are $\bfR$-orthogonal, that is, $(\bfX^{(1)})^\rmT \bfL \bfX^{(2)}=\bf0_{K,K}$ and
 $\bfX^{(1)} \bfR (\bfX^{(2)})^\rmT =\bf0_{L,L}$.
\end{definition}

\begin{definition}
Two series $\tX^{(1)}_N$ and $\tX^{(2)}_N$ are called \emph{strongly $(\bfL,\bfR)$-separable}, if they are weakly $(\bfL,\bfR)$-separable
and $\sigma_{1,i} \neq \sigma_{2,j}$ for any $i$ and $j$.
\end{definition}

The $(\bfL,\bfR)$-separability of two series components means $\bfL$-orthogonality of their
subseries of length $L$ and $\bfR$-orthogonality of the
subseries of length $K=N-L+1$.

The following theorem shows that the $(\bfL,\bfR)$-separability is in a sense much less restrictive
than the ordinary one.

\begin{theorem}
\label{th:sepSVD}
Let $\tX=\tX^{(1)}+\tX^{(2)}$ be the series of length $N$, $L$ be the window length and the $L$-rank of $\tX$ be equal to $r$.
Let $\tX^{(m)}$ be the series of $L$-rank $r_m$, $m=1,2$, $r_1+r_2=r$.
Then there exist separating matrices $\bfL\in \cM_{L,L}$ and
$\bfR\in \cM_{K,K}$ of rank $r$ such that the series $\tX^{(1)}$ and $\tX^{(2)}$ are strongly $(\bfL,\bfR)$-separable.
\end{theorem}
\begin{proof}
Denote $\{P_i^{(m)}\}_{i=1}^{r_m}$ a basis of the column space of $\bfX^{(m)}$ and
$\{Q_i^{(m)}\}_{i=1}^{r_m}$ a basis of the row space of $\bfX^{(m)}$, $m=1,2$; e.g.,
$P_i^{(m)}=P_{m,i}\in \spaceR^L$, $Q_i^{(m)}=Q_{m,i}\in \spaceR^K$.
Define
\begin{eqnarray*}
\bfP=[P_1^{(1)}:\ldots:P_{r_1}^{(1)}:P_1^{(2)}:\ldots:P_{r_2}^{(2)}],\\
\bfQ=[Q_1^{(1)}:\ldots:Q_{r_1}^{(1)}:Q_1^{(2)}:\ldots:Q_{r_2}^{(2)}].
\end{eqnarray*}
By the theorem conditions, the matrices $\bfP$ and $\bfQ$ are of full rank.
Since $\bfP^\dag$ and $\bfQ^\dag$ orthonormalize the columns of the matrices $\bfP$ and $\bfQ$ (Proposition~\ref{prop:orth}),
then the trajectory matrices $\bfX^{(1)}$ and $\bfX^{(2)}$ are $(\bfL,\bfR)$ bi-orthogonal for $\bfL=(\bfP^\dag)^\rmT \bfP^\dag$
and $\bfR=(\bfQ^\dag)^\rmT \bfQ^\dag$.
Therefore the series $\tX^{(1)}$ and $\tX^{(2)}$ are $(\bfL,\bfR)$-separable.

Proposition~\ref{prop:reord} shows that it is possible to change $\sigma_{m,i}$ keeping bi-orthogonality,
that is, it explains how to get strong separability not corrupting weak one.
\end{proof}

\begin{remark}
\label{rem:roots}
  Consider two time series governed by minimal LRRs of orders $r_1$ and $r_2$, $r_1+r_2 \leq \min(L,K)$.
  The conditions of Theorem~\ref{th:sepSVD} fulfill if and only if the sets of characteristic roots of the series are disjoint.
  Really, the sets of characteristic roots are disjoint if and only if the column and row spaces of $L$-trajectory matrices intersect only in $\{0\}$,
  that is,  $\bfP$ and $\bfQ$ are of full rank.
\end{remark}

\begin{remark}
Theorem~\ref{th:sepSVD} together with Remark~\ref{rem:roots} shows that any two times series
governed by LRRs with different characteristic roots can be separated by some $(\bfL,\bfR)$-SVD
for sufficiently large series and window lengths.
\end{remark}

Note that
Theorem~\ref{th:sepSVD} is not constructive, since the trajectory spaces of the separated series should be known
for exact separation.
However, we can try to estimate these spaces and thereby to improve the separability.

\smallskip
{\bf Measures of oblique separability}.
If Oblique SSA does not separate the components exactly, a measure of separability
is necessary. We can consider the analogue of $\bfw$-correlations described in Section~\ref{sec:sep_mes}, since they are defined through
the Frobenius inner products of trajectory matrices and therefore can be generalized; see Appendix~\ref{sec: minnerprod}
for definition of $\rho_{\bfL,\bfR}$ in \eqref{eq:ABcorr}.
Define $(\bfL,\bfR)$ $\bfw$-correlation between the reconstructed series $\tilde\tX^{(1)}$ and $\tilde\tX^{(2)}$
as $\rho_{(\bfL,\bfR)}(\tilde\bfX^{(1)}, \tilde\bfX^{(2)})$.
Note that due to diagonal averaging, the column and row spaces of $\tilde\bfX^{(m)}$
do not necessarily belong to the column spaces of $\bfL$  and $\bfR$ correspondingly,
that is, matrices $\bfL$ and $\bfR$ can be not consistent with $\tilde\bfX^{(m)}$, $m=1,2$.
Therefore, $\rho_{\bfL,\bfR}$ takes into consideration only projections of columns and rows of $\tilde\bfX^{(1)}$ and $\tilde\bfX^{(2)}$
on the column spaces of $\bfL$ and $\bfR$ (Remark~\ref{rem:noagree}). This means that $\rho_{\bfL,\bfR}$ can overestimate
the separability accuracy.

For Oblique SSA, when only one of coordinate systems (left or right) is oblique, the conventional $\bfw$-correlations between series are  more appropriate measures of separability,
since in the case of exact oblique separability we have orthogonal (in the Frobenius inner product) matrix components (Corollary~\ref{col:F_orth}).

Other important measure of proper separability is the closeness of the reconstructed series components
to time series of finite rank.
This can be measured by the contribution of the leading $r_m=|I_m|$ eigentriples into the SVD of
the trajectory matrix $\wtilde\bfX^{(m)}$ of the $m$th reconstructed series component $\wtilde\tX^{(m)}$.
If we denote $\wtilde\lambda_{m,i}$ the eigenvalues of the ordinary SVD of  $\wtilde\bfX^{(m)}$, then
$\tau_{r_m}(\wtilde\tX^{(m)})=1-\sum_{i=1}^{r_m} \wtilde\lambda_{m,i}/\|\wtilde\bfX^{(m)}\|^2$ reflects the closeness of the $m$th series to the series
of rank $r_m$.

\subsection{Nested Oblique SSA}
\label{sec:nested}
Rather than the ordinary SVD, the SVD with respect to non-orthogonal coordinate systems provides approximation in an inappropriate way.
That is why Oblique SSA cannot
be used for extraction of the leading components, in particular, for extraction
of the signal and for denoising.

Therefore, the nested way of using Oblique SSA is suggested. The approach is somewhat similar to factor analysis,
where a factor space can be estimated by principal component analysis and then interpretable factors are extracted from the factor space.

Suppose that Basic SSA can extract the signal but cannot separate the signal components.
For example, let the time series consist of a noisy sum of two sinusoids.
Then Basic SSA can perform denoising but probably cannot separate these sinusoids,
if their frequencies are close.
Thus, Basic SSA is used for estimation of the subspace of the sum of sinusoids
and then some other method can be used to separate the sinusoids themselves.
The choice of parameters for better separation is thoroughly investigated in
\cite{Golyandina2010}.
Note that the nested approach is similar to the refined grouping used in \cite[Section 2.5.4]{Golyandina.Zhigljavsky2012}
for the SSA-ICA algorithm.

Thus, let us apply Basic SSA with proper parameters and let a matrix decomposition $\bfX = \bfX_{I_1} + \ldots + \bfX_{I_p}$ be
obtained at Grouping step of Basic SSA; each group  corresponds
to a separated time series component. Let the $s$th group $I=I_s$ be chosen for a refined decomposition.
Denote  $\bfY=\bfX_{I}$, $r=\rank \bfY$, $\tY=\calT^{-1}\calH \bfY$ the series obtained from $\bfY$ by diagonal averaging.

\begin{algorithm}
\label{alg:nested}
\rm
\textbf{(Nested Oblique SSA.)}

\smallskip\noindent
\emph{Input}: The matrix $\bfY$, matrices  $(\bfL, \bfR)$, which are consistent with $\bfY$ (see Definition~\ref{def:agreed}).

\smallskip\noindent
\emph{Output}: a refined series decomposition $\tY=\wtilde\tY^{(1)} + \ldots + \wtilde\tY^{(l)}$.

\begin{enumerate}
\item
Construct an $(\bfL,\bfR)$-SVD of $\bfY$ by Algorithm~\ref{alg:LRSVD} in the form
\bea
\label{eq:nestedSVD}
\bfY=\sum_{i=1}^r \sigma_i P_i Q_i^\rmT.
\eea
\item
Partition the set $\{1,\ldots,r\}=\bigsqcup_{m=1}^l J_m$ 
and perform grouping to obtain a refined matrix decomposition
$\bfY=\bfY_{J_1} + \ldots + \bfY_{J_l}$.
\item
Obtain a refined series decomposition
$\tY=\wtilde\tY^{(1)} + \ldots + \wtilde\tY^{(l)}$, where
$\wtilde\tY^{(m)}=\calT^{-1}\calH \bfY_{J_m}$.
\end{enumerate}
\end{algorithm}

Thus, after application of Algorithm~\ref{alg:nested} to the group $I_s$, we obtain the following decomposition of the series $\tX$:
\bea
\tX=\wtilde\tX^{(1)}+\ldots+\wtilde\tX^{(p)}, \ \mbox{where}\ \wtilde\tX^{(s)}=\wtilde\tY^{(1)} + \ldots + \wtilde\tY^{(l)}.
\eea

For simplicity, below we will consider the case $l=2$.

\subsection{Iterative O-SSA}

Let us describe an iterative version of Algorithm~\ref{alg:nested}, that is,
an iterative algorithm for obtaininig appropriate matrices $\bfL$ and $\bfR$ for the $(\bfL, \bfR)$-SVD of $\bfX_I$.
For proper use of nested decompositions, we should
expect that the matrix $\bfX_I$ is close to a rank-deficient trajectory matrix of rank $r$.

To explain the main point of the method, assume that $\bfX_I=\bfY$ is the trajectory matrix of $\tY$.
Let $\tY=\tY^{(1)}+\tY^{(2)}$ and the trajectory matrices $\bfY_1$ and $\bfY_2$ be of ranks $r_1$ and $r_2$, $r_1+r_2=r$. Then
by Theorem~\ref{th:sepSVD} there exist $r$-rank separating matrices $\bfL^*$, $\bfR^*$ of sizes $L\times L$ and
$K\times K$ correspondingly and a partition
$\{1,\ldots,r\}=J_1 \sqcup J_2$ such
that we can perform the proper grouping in the $(\bfL^*,\bfR^*)$-SVD and thereby obtain $\bfY_{J_1}=\bfY_1$ and $\bfY_{J_2}=\bfY_2$.

Unfortunately, we do not know $\bfL^*$ and $\bfR^*$, since they are determined
by unknown trajectory spaces of $\tY^{(1)}$ and $\tY^{(2)}$.
Therefore, we want to construct the sequence of $(\bfL,\bfR)$-SVD decompositions \eqref{eq:LRSVD},
which in some sense converges to the separating decomposition.

Let us have an initial $(\bfL^{(0)},\bfR^{(0)})$-SVD decomposition of $\bfY$ and group its
components to obtain some initial estimates $\wtilde\tY^{(1,0)}$ and $\wtilde\tY^{(2,0)}$ of $\tY^{(1)}$ and $\tY^{(2)}$.
Then we can use the trajectory spaces of $\wtilde\tY^{(1,0)}$ and $\wtilde\tY^{(2,0)}$ to construct the new inner product
which is expected to be closer to the separating one. Therefore, we can expect that $\wtilde\tY^{(1,1)}$ and $\wtilde\tY^{(2,1)}$
will be closer to $\tY^{(1)}$ and $\tY^{(2)}$ and therefore we take their trajectory spaces
to construct a new inner product; and so on.
Certainly, if the initial decomposition is strongly separating, then we obtain that
$\wtilde\tY^{(m,1)}=\wtilde\tY^{(m,0)}=\tY^{(m)}$, $m=1,2$.


\subsubsection{Basic algorithm}

We call the iterative version of Algorithm~\ref{alg:nested} (Nested Oblique SSA) as \emph{Iterative Oblique SSA} or \emph{Iterative O-SSA}.

\begin{algorithm}
\label{alg:iter}
\rm
\textbf{(Scheme of Iterative O-SSA.)}

\smallskip\noindent
\emph{Input:}
The matrix $\bfY$ of rank $r$, which is the input matrix for Algorithm~\ref{alg:nested},
a partition $\{1,\ldots,r\}=J_1 \sqcup J_2$, $r_m=|J_m|$, the accuracy $\ve$ and the maximal number of iterations $M$.
Also we should choose a pair of matrices $(\bfL^{(0)},\bfR^{(0)})$, consistent with $\bfY$ as initial
data.

\smallskip
Together with the partition, the matrices provide the decompositions $\bfY= \bfY_{J_1}^{(0)} + \bfY_{J_2}^{(0)}$
and $\tY=\wtilde\tY^{(1,0)} + \wtilde\tY^{(2,0)}$.

\smallskip\noindent
\emph{Output:}
$\tY=\wtilde\tY^{(1)} + \wtilde\tY^{(2)}$.

\begin{enumerate}
\item
$k=1$.
\item
\label{step:LkRk}
Call of Algorithm for calculation of $(\bfL^{(k)},\bfR^{(k)})$ consistent with $\bfY$.
\item
Construct the $(\bfL^{(k)},\bfR^{(k)})$-SVD of $\bfY$:
\begin{equation}
\label{eq:kLRSVD}
  \bfY=\sum_{i=1}^r \sigma_i^{(k)} P_i^{(k)} (Q_i^{(k)})^\rmT = \bfY_{J_1}^{(k)} + \bfY_{J_2}^{(k)}.
\end{equation}
\item
Obtain the decomposition of the series $\tY=\wtilde\tY^{(1,k)} + \wtilde\tY^{(2,k)}$,
where $\wtilde\tY^{(m,k)}=\calT^{-1}\calH \bfY_{J_m}^{(k)}$, $m=1,2$.
\item
If $k\ge M$ or $\max(\|\wtilde\tY^{(m,k)}-\wtilde\tY^{(m,k-1)}\|^2/N, m=1,2)<\varepsilon^2$, then
$\wtilde\tY^{(m)} \leftarrow \wtilde\tY^{(m,k)}$, $m=1,2$,
and STOP;
else $k \leftarrow k+1$ and go to step 2.
\end{enumerate}
\end{algorithm}

\begin{remark}
\label{rem:L0R0}
Note that the initial matrices $(\bfL^{(0)},\bfR^{(0)})$ can be chosen such that the initial decomposition (\ref{eq:kLRSVD}) for $k=0$ is a part
of the SVD \eqref{eq:elem_matr} and thereby coincides with the ordinary SVD of $\bfY$, that is, $\bfL^{(0)}$ and $\bfR^{(0)}$ are the identity matrices.
Then the partition can be performed as follows. In the decomposition \eqref{eq:elem_matr},
we can choose two sets of eigentriple numbers and consider their union as $I$. The chosen sets of
numbers automatically generate the partition $J_1\sqcup J_2$. For example, if two groups,
ET2,8 and ET3--6, are chosen, then $I=\{2,3,4,5,6,8\}$, $r=6$, $J_1=\{1,6\}$, $J_2=\{2,3,4,5\}$.
\end{remark}

To finalize Algorithm~\ref{alg:iter}, we should present the algorithm for step~\ref{step:LkRk}.
Define
$\Pi_\mathrm{col}$ the orthogonal projection operator (in the ordinary sense) on the column space
of $\bfY$, $\Pi_\mathrm{row}$ the projection operator on the row space of $\bfY$.

\begin{algorithm}
\label{alg:LRcalc}
\rm
\textbf{(Calculation of $(\bfL^{(k)},\bfR^{(k)})$.)}

\smallskip\noindent
\emph{Input}:
The partition $\{1,\ldots,r\}=J_1 \sqcup J_2$, $r_m=|J_m|$,
the pair of matrices $(\bfL^{(k-1)},\bfR^{(k-1)})$.

\smallskip\noindent
\emph{Output}:
The pair of matrices $(\bfL^{(k)},\bfR^{(k)})$.

\begin{enumerate}
\item
\label{it:hankel}
Calculate $\wtilde\bfY_m=\calH \bfY_{J_m}^{(k-1)}$, $m=1,2$.
\item
\label{it:2SVD}
Construct the ordinary SVDs:
\bea
\wtilde\bfY_m=\sum_{i=1}^{d_m} \sqrt{\lambda_i^{(m)}} U_i^{(m)} (V_i^{(m)})^\rmT,\ m=1,2,
\eea
(we need the first $r_m$ terms only).
\item
\label{it:whatU}
Find the projections $\what{U}_i^{(m)} = \Pi_\mathrm{col} U_i^{(m)}$ and $\what{V}_i^{(m)} = \Pi_\mathrm{row} V_i^{(m)}$ for $i=1,\ldots, r_m$,  $m=1,2$.
Denote
\bea
\what{\bfU}^{(m)}=[\what{U}_1^{(m)}:\ldots:\what{U}_{r_m}^{(m)}], \ \what{\bfV}^{(m)}=[\what{V}_1^{(m)}:\ldots:\what{V}_{r_m}^{(m)}].
\eea
\item
\label{it:laststep}
Calculate $\bfL^{(k)}=(\what\bfU^\dag)^\rmT \what\bfU^\dag$ and
$\bfR^{(k)}=(\what\bfV^\dag)^\rmT \what\bfV^\dag$, where
$\what\bfU=[\what{\bfU}^{(1)}:\what{\bfU}^{(2)}]$  and $\what\bfV=[\what{\bfV}^{(1)}:\what{\bfV}^{(2)}]$.
\end{enumerate}
\end{algorithm}


Note that we assume that the matrices $\what{\bfU}^{(m)}$ and $\what{\bfV}^{(m)}$ obtained at step~\ref{it:whatU} are
of full rank; otherwise, the algorithm does not work.

For the constructed iterative Algorithm~\ref{alg:iter}, the convergence of $\wtilde\tY^{(1,k)}$ and $\wtilde\tY^{(2,k)}$
to some  series $\tY^{(1)}$ and $\tY^{(2)}$ is not proved theoretically;
however, numerical experiments confirm the convergence for the most of reasonable examples.

Let us shortly discuss why one can expect the convergence of the iterations to the proper decomposition.
First, note that Iterative O-SSA does not change the separating decomposition, that is, the separating decomposition
is a fixed point of the algorithm.
Then, the separating decomposition $\bfY = \bfY_1 + \bfY_2$ should satisfy the following properties:\\
(1) $\bfY_1$ and $\bfY_2$ are Hankel;\\ (2) $\rank \bfY_1 =r_1$, $\rank\bfY_2 =r_2$; \\(3) the column and row spaces of $\bfY_1$ and $\bfY_2$ lie in the column and row spaces of $\bfY$;\\ (4) $\bfY_1$ and $\bfY_2$ are $(\bfL,\bfR)$ bi-orthogonal for $\bfL=\bfL^*$ and $\bfR=\bfR^*$.

Each iteration consequently tries to meet these properties:\\
(1) hankelization at step~\ref{it:hankel} is the orthogonal projection on the set
of Hankel matrices;\\
(2) taking the $r_m$ leading components in the SVDs of series (step~\ref{it:2SVD}) performs the low-rank projections;\\
(3) there is the step~\ref{it:whatU} of projection on the row and column spaces of $\bfY$;\\
(4) the choice of $(\bfL,\bfR)$-inner products at step~\ref{it:laststep} makes the matrices bi-orthogonal.

\subsubsection{Modification with sigma-correction}
If the initial point for iterations is not far from the separating pair $(\bfL^*$, $\bfR^*)$,
we can expect that the convergence will take place, since we are close to the fixed-point value
and we can expect that $\sigma_i^{(k)}$ are changed slightly.
However, in general, a possible reordering of the decomposition components
between iterations of Iterative O-SSA can interfere convergence.
The case of $J_1=\{1,\ldots,r_1\}$, when the minimal singular value $\sigma_{r_1}$ of the first series is kept essentially larger than the maximal singular value $\sigma_{r_1+1}$ of the second series,
would provide safety.

Let us describe the modification of Iterative O-SSA that provides reordering of the components, moves them apart
and thereby relaxes the problem of mixing of components.
Modification consists in an adjustment of calculation of $\what{\bfU}^{(2)}$ and $\what{\bfV}^{(2)}$\ at step \ref{it:whatU}
of Algorithm~\ref{alg:LRcalc}.

\begin{algorithm}
\label{alg:sigmacor}
\rm
\textbf{(Modification of Algorithm~\ref{alg:LRcalc}.)}

\smallskip\noindent
\emph{Input} and \emph{Output} are the same as in Algorithm~\ref{alg:LRcalc} except for
an additional parameter $\varkappa>1$ called the separating factor.

\smallskip\noindent
The algorithm is the same except for an additional step 3a after step 3.

\smallskip\noindent
3a: If $\lambda_{r_1}^{(1)} < \varkappa^2 \lambda_1^{(2)}$ at step \ref{it:2SVD} of Algorithm~\ref{alg:LRcalc}, then
define $\mu = \varkappa \sqrt{\lambda_1^{(2)} / \lambda_{r_1}^{(1)}}$ and change
$\what{\bfU}^{(2)} \leftarrow \sqrt{\mu}\what{\bfU}^{(2)}$, $\what{\bfV}^{(2)} \leftarrow \sqrt{\mu}\what{\bfV}^{(2)}$.
Due to reordering, put $J_1=\{1,\ldots,r_1\}$, $J_2=\{r_1+1,\ldots,r\}$.
\end{algorithm}

Note that the adjustment implicitly leads to the change of the order of matrix components in \eqref{eq:kLRSVD},
since they are ordered by $\sigma^{(k)}_i$. Thereby
we force an increase of the matrix components related to the first series component.
Proposition~\ref{prop:reord} explains this adjustment.


\begin{remark}
The reordering procedure is fulfilled by sequential adjustment of
the component weights and therefore depends on the component enumeration.
\end{remark}

Note that the described correction can help to provide the strong separability
if the weak one takes place.

\subsection{Separability of sine waves with close frequencies}
\label{sec:close_freq}
\subsubsection{Noise-free cases}
Let us consider the sum of two sinusoids $x_n=\sin(2\pi \omega_1 n)+ A\sin(2\pi \omega_2 n)$, $n=1,\ldots, N$, $N=150$,
with close frequencies $\omega_1 = 0.065$ and $\omega_2 = 0.06$ and unequal amplitudes, 1 and $A=1.2$. Let the window length $L=70$.
Since sinusoids with such close frequencies are far from being orthogonal for the considered window and series lengths,
Basic SSA cannot separate them, see Fig.~\ref{fig:assa1} (top) where the result of the Basic SSA decomposition is depicted.

To separate the sinusoids we apply the Iterative O-SSA algorithm (Algorithm~\ref{alg:iter}) with no sigma-correction,
$\ve=10^{-5}$ and two groups ET1--2 and ET3--4. The maximal number $M$ of iterations was taken very large and therefore
was not reached.
Decomposition after Iterative O-SSA is depicted in Fig.~\ref{fig:assa1} (bottom).

\begin{figure}[!htb]
    \centering
    \includegraphics{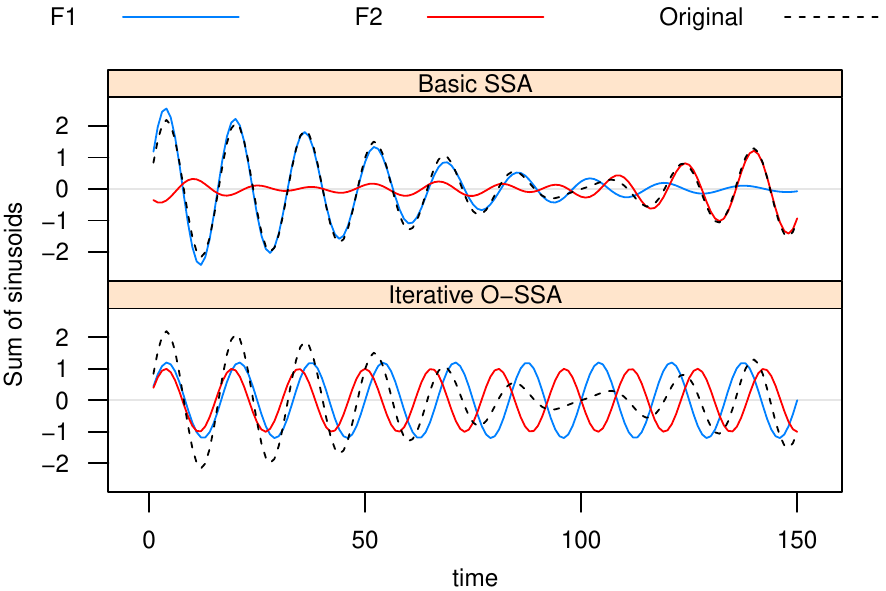}
    \caption{Sum of two sinusoids with close frequencies: decomposition by Basic SSA and Iterative O-SSA}
    \label{fig:assa1}
\end{figure}

Let us apply the measures of separability described in Section~\ref{sec:LRsepar}.
Note that the conventional $\bfw$-correlations do not reflect the quality of decomposition.
For the initial decomposition we have $0.08$. After Iterative O-SSA the $\bfw$-correlation becomes to be equal to $-0.44$, while
$(\bfL,\bfR)$ $\bfw$-correlation is almost 0. The last result confirms that
the method separates harmonics exactly.
Other measure of true decomposition is the closeness of the components to series of finite ranks.
Since the ranks should be equal to the number of the components in the chosen groups, we can calculate
the proportion of the corresponding number of the leading components in their SVD decompositions.
The mean proportion ($0.5(\tau_{r_1}(\tX^{(1)})+\tau_{r_2}(\tX^{(2)})$) is changed from $0.06$ to almost 0.

Let us fix $\omega_2=0.06$. Then for $\omega_1=0.065$ the algorithm stops after 113 iterations,
for $\omega_1=0.07$ the number of iterations is equal to 26, for $\omega_1=0.08$ it is equal to just 6;
see blue line in Fig.~\ref{fig:iter_err} (top).

Note that we do not need to use the sigma-correction, since
the sinusoids have different amplitudes.

If we consider equal amplitudes with $A=1$ and take $\varkappa=2$ (Algorithm~\ref{alg:sigmacor}),
 then Iterative O-SSA still converges even for $\omega_2=0.065$ (191 iterations) to the true solution.

\subsubsection{Nested separability in presence of noise}
Let us add noise to the sum of two sinusoids and take
$x_n=\sin(2\pi \omega_1 n)+ A\sin(2\pi \omega_2 n) + \delta \ve_n$
with close frequencies $\omega_1 = 0.07$ and $\omega_2 = 0.06$ and unequal amplitudes, 1 and $A=1.2$.
Here $\ve_n$ is white Gaussian noise with variance 1, $\delta=1$. Let again $N=150$, $L=70$.

\begin{figure}[!htb]
    \centering
    \includegraphics{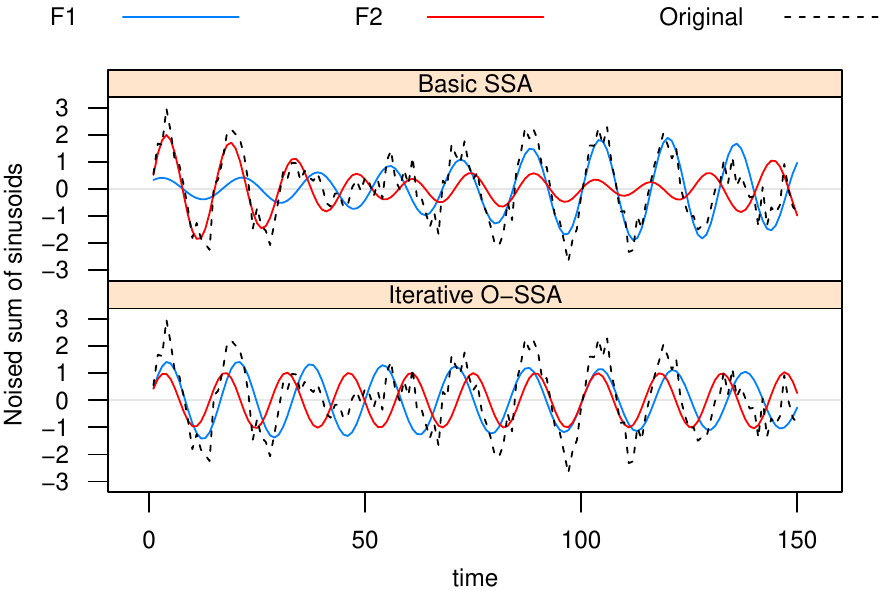}
    \caption{Noisy sum of two sinusoids with close frequencies: decomposition by Basic SSA and Iterative O-SSA}
    \label{fig:assa2}
\end{figure}

Basic SSA well separates the sinusoids from noise, but cannot separate these sinusoids themselves.
Thus, Iterative O-SSA applied to the estimated signal subspace should be used.
We use the sigma-correction with $\varkappa=2$, since the difference between amplitudes, 1 and 1.2,
appears to be small for strong separability in presence of noise. As before, we set the initial grouping
ET1--2 and ET3--4.

\begin{figure}[!htb]
    \centering
    \includegraphics{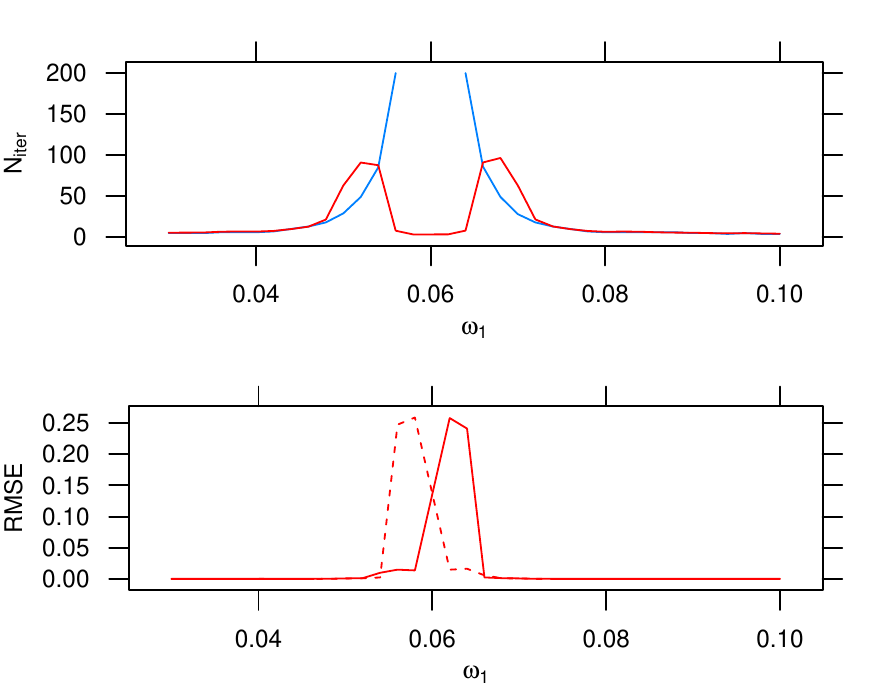}
    \caption{Dependence of number of iterations (top) and RMSE errors of frequency estimations (bottom) on $\omega_1$ for $\omega_2=0.6$}
    \label{fig:iter_err}
\end{figure}

The decomposition by Basic SSA at top and by Iterative O-SSA at bottom is depicted in Fig.~\ref{fig:assa2}.
The number of iterations is equal to 32, what is just slightly larger than 26 in the noiseless case.

Let us investigate the dependence of number of iterations on $\omega_1$ with the fixed $\omega_2=0.06$.
We change $\omega_1$ from 0.03 to 0.059 and from 0.061 to 0.1.
Fig.~\ref{fig:iter_err} (top) shows the number of iterations for noiseless signal
(blue line) and the estimated mean number of iterations for the noisy signal (red line);
the number of repetitions equals 1000, 5\% winsorized estimates of means were calculated.
Note that the number of iterations was limited by 200, although
for the pure signal convergence held for each $\omega_1$ from the considered set.
A surprisingly small number of iterations for the noisy signal
and close frequencies is explained by convergence to an wrong limit, see Fig.~\ref{fig:iter_err} (bottom) with root mean square errors
of LS-ESPRIT estimates for $\omega_1$ and $\omega_2$ based on the subspaces spanned by eigenvectors from ET1--2 and ET3--4
(see, e.g., \cite{Roy.Kailath1989} or \cite[Section 2.8.2]{Golyandina.Zhigljavsky2012} for the ESPRIT algorithms).
Since we use the nested decomposition, the noise slightly influences the reconstruction accuracy for frequencies
that are quite different ($\omega_1$ smaller than 0.048 and larger than 0.072).

\section{SSA with derivatives. Variation for strong separability}
\label{sec:DerivSSA}
In this section we describe a variation of SSA that helps to overcome the problem of lack
of strong separability if weak separability holds.

Recall that the lack of strong separability of two series components is caused by equal singular values in the sets of the singular values generated by each of time series.
In turn, the singular values depends on coefficients $A_1$ and $A_2$ before the series components in the sum $A_1 s_n^{(1)} + A_2 s_n^{(2)}$.
The question is how to change the coefficients $A_1$ and $A_2$ in conditions of unknown $s_n^{(1)}$ and $s_n^{(2)}$
to make the singular values different.

It seems that the most natural approach is to use the derivative of the time series in order to change the coefficients
and not to change the component subspaces.
For example, if $x_n = A \sin(2 \pi \omega n + \phi)$, then $x'_n =  2\pi \omega A\cos(2 \pi \omega n +\phi)$,
that is, the coefficient $A'=2\pi \omega A$.
If we take two sinusoids with different frequencies, then derivation changes their amplitudes
differently.
For $x_n= A e^{\alpha n}$, derivation also changes the coefficient before the exponential, since $x'_n=\alpha A e^{\alpha n}$,
and preserves the rate.
For the most of series of finite rank, the derivative subspace coincides with the series subspace.
The exception is polynomial series, when the derivative subspace is a subset
of the initial subspace.

Certainly, since we deal with discrete time,
we consider $\varphi_n(\tX)=x_{n+1}-x_{n}$ instead of derivative. However, the approach of taking differences works.
For example, for series $\tX=\tX_N$ of length $N$ with $x_n = A \sin(2 \pi \omega n + \phi)$, we obtain the series
$\Phi_{N-1}(\tX)=(\varphi_1(\tX),\ldots,\varphi_{N-1}(\tX))$ of length $N-1$ with
$\varphi_n(\tX) =  2 \sin(\pi \omega) A\cos(2 \pi \omega n + \pi \omega +\phi)$;
for $x_n = A e^{\alpha n}$, we obtain $\varphi_n(\tX)= (e^{\alpha} -1) A e^{\alpha n}$.


Thus, we can combine the initial series and
its derivative to imbalance the component contribution and therefore to obtain their strong separability.
For sinusoids, the smaller the period, the larger the increase of the
sinusoid amplitude. Therefore, derivation increases the contribution of
high frequencies. This effect can increase the level of the noise component, if the series is corrupted by noise.
Hence, the nested version of the method implementation should be produced; in particular, the noise component should be
removed by Basic SSA in advance.

\begin{remark}
The approach involving derivatives (that is, sequential differences) can be naturally extended to considering
an arbitrary linear filtration $\varphi$ instead of taking sequential differences. It this paper we deal with derivatives, since
this particular case is simple and has very useful applications.
\end{remark}

In Section~\ref{sec:FSSA_MSSA} we consider the initial series and its derivative together as two series,
regulating the contribution of the derivative, and apply then the multivariate version of SSA.
Section~\ref{sec:FSSA_DSSA} transforms this approach to a special nested version of Oblique SSA called DerivSSA.

\subsection{SSA with derivatives as MSSA}
\label{sec:FSSA_MSSA}
Let us consider the system of two time series $(\tX_N,\gamma \Phi_{N-1}(\tX))$
and apply Multivariate SSA (MSSA).

The MSSA algorithm can be found, for example, in \cite{Elsner.Tsonis1996, Golyandina.Stepanov2005} for time series of equal lengths.
However, it is naturally extended to different lengths.
In particular, MSSA for time series of different lengths is described in \cite[Section III.2]{Danilov.Zhigljavsky1997} and
\cite{Golyandina.etal2013}.

In MSSA, the embedding operator $\calT$ transfers two time series
$(\tX_{N_1},\tY_{N_2})$ to the stacked $L$-trajectory matrix $[\bfX:\bfY]$.
That is, the only difference with Basic SSA consists in the construction of the embedding operator $\calT$.

Let $\tX_N=\tX_N^{(1)}+\tX_N^{(2)}$ and $\tX_N^{(1)}$ and $\tX_N^{(2)}$ be of finite rank and approximately separable.
Therefore their row and column trajectory spaces are approximately orthogonal.
Then the same is valid for $\Phi_{N-1}(\tX^{(1)})$ and $\Phi_{N-1}(\tX^{(2)})$ in view of the fact that
their column spaces belongs to the column spaces of $\tX_N^{(1)}$ and $\tX_N^{(2)}$, while
their row spaces are spanned by vectors of the same structure that the vectors constituting bases of the row spaces of $\tX_N^{(1)}$ and $\tX_N^{(2)}$,
except for these basis vectors has length $K-1$ instead of $K$. Therefore, approximate orthogonality still hold.
Since $\Phi_{N-1}(\tX)=\Phi_{N-1}(\tX^{(1)})+\Phi_{N-1}(\tX^{(2)})$,
MSSA applied to $(\tX_N,\gamma \Phi_{N-1}(\tX))$ will approximately separate the time series $\tX_N^{(1)}$ and $\tX_N^{(2)}$.
Certainly, we will not have exact separability; however, it is not so important for practice.

As it was mentioned before, a drawback of the described approach is that the method cannot be  applied to noisy series, since
it intensifies high-frequency harmonics and therefore strengthens noise.
Therefore, denoising should be applied as preprocessing.
Also, SSA involving derivatives changes component contributions (this is what we want) but
simultaneously the method loses approximation features.
These reasons lead to the necessity to use the nested way of decomposition introduced in Section~\ref{sec:nested}.

\subsection{Nested SSA with derivatives (DerivSSA)}
\label{sec:FSSA_DSSA}
Let us formulate the nested version of SSA with derivatives called DerivSSA.
As well as in Section~\ref{sec:nested}, let $\bfY=\bfX_{I}$ be one of matrices in the decomposition
$\bfX = \bfX_{I_1} + \ldots + \bfX_{I_p}$ obtained at Grouping step of Basic SSA; each group  corresponds
to a separated time series component and we want to construct a refined decomposition of $\bfY$.
As before, denote  $r=\rank \bfY$, $\tY=\calT^{-1}\calH \bfY$.

\begin{algorithm}
\label{alg:deriv}
\rm
\textbf{(DerivSSA.)}

\smallskip\noindent
\emph{Input}: The matrix $\bfY$, the weight of derivative $\gamma>0$.

\smallskip\noindent
\emph{Output}: a refined series decomposition $\tY=\wtilde\tY^{(1)} + \ldots + \wtilde\tY^{(l)}$.

\begin{enumerate}
\item
Denote $\Phi(\bfY)=[Y_2-Y_1:\ldots:Y_K-Y_{K-1}]$.
Construct the matrix $\bfZ=[\bfY:\gamma\Phi(\bfY)]$.
\item
Perform the SVD of $\bfZ$: $\bfZ=\sum_{i=1}^r \sqrt{\lambda_i} U_i V_i^\rmT$.
\item
Construct the following decomposition of $\bfY=\bfX_I$ into the sum
of elementary matrices: $\bfY=\sum_{i=1}^r U_i U_i^\rmT \bfY$.
\item
Partition the set $\{1,\ldots,r\}=\bigsqcup_{m=1}^l J_m$ 
and perform grouping to obtain a refined matrix decomposition
$\bfY=\bfY_{J_1} + \ldots + \bfY_{J_l}$.
\item
Obtain a refined series decomposition
$\tY=\wtilde\tY^{(1)} + \ldots + \wtilde\tY^{(l)}$, where
$\wtilde\tY^{(m)}=\calT^{-1}\calH \bfY_{J_m}$.
\end{enumerate}
\end{algorithm}

Note that steps 2 and 3 of algorithm are correct, since the column space of $\bfZ$ coincides with
the column space of $\bfY$. Therefore, $\rank \bfZ=r$ and $\{U_i\}_{i=1}^r$ is the orthonormal basis of
the column space of $\bfY$.

The following proposition shows that Algorithm~\ref{alg:deriv} is exactly Algorithm~\ref{alg:nested} with a specific pair of matrices
$(\bfL,\bfR)$, where
$P_i=U_i$, $Q_i$ are normalized vectors $\bfY^\rmT U_i$ in \eqref{eq:nestedSVD}.

\begin{proposition}
\label{prop: DSSA_OSSA}
  The left singular vectors of the ordinary SVD of $\bfZ$ coincide with
  the left singular vectors of the ($\bfL$,$\bfR$)-SVD of the input matrix  $\bfY$,
  where $\bfL\in \cM_{L,L}$  is the identity matrix and $\bfR$ is defined by the equality
  $\bfR = \bfE + \gamma^2 \bfF^\rmT \bfF$, where $\bfE\in \cM_{K,K}$ is the identity matrix and
  \begin{gather*}
  \bfF = \left(\begin{array}{cccccc}
      -1 & 1 & 0 & 0 & \cdots & 0 \\
      0 & -1 & 1 & 0 & \cdots & 0 \\
      \vdots & \ddots & \ddots & \ddots & \ddots & \vdots \\
      0 & \cdots & 0 & -1 & 1 & 0\\
      0 & \cdots & 0 & 0 & -1 & 1\\
    \end{array}\right)
    \in \cM_{K-1,K}.
\end{gather*}
\end{proposition}
\begin{proof}
 Note that the standard inner product in the row space of $\bfZ$
 can be expressed as $( Z_1, Z_2 )_{2K-1} = (Q_1,Q_2)_K + \gamma^2(\Phi(Q_1),\Phi(Q_2))_{K-1}$,
 where $Q_1$ and $Q_2$ consist of the first $K$ components of $Z_1$ and $Z_2$,
 $\Phi(Q)\in \spaceR^{K-1}$ applied to a vector $Q=(q_1,\ldots,q_K)^\rmT$
  consists of successive differences of vector components $q_{i+1}-q_i$.
 Thus, if we introduce the inner product $\langle Q_1, Q_2 \rangle_\bfR= (\bfR Q_1,Q_2)_K$, then
 the ordinary SVD of $\bfZ$ can be reduced to the ($\bfL$,$\bfR$)-SVD of $\bfY$ with
 the corresponding matrices $\bfL$ and $\bfR$.
\end{proof}

\begin{remark}
\label{prop:MSSA_DSSA}
If $\bfY$ is the trajectory  matrix of a series $\tY_N$, then the nested SSA with derivatives is equivalent
to the MSSA implementation described in Section~\ref{sec:FSSA_MSSA}. Indeed, the trajectory matrix of the derivative time series $\Phi_{N-1}(\tY)$
coincides with the matrix $\Phi(\bfY)$. Although, if $\bfY$ is not Hankel, there is no MSSA analogue.
\end{remark}


\subsection{Separation of sine waves with equal amplitudes}
\label{sec:close_ampl}
Consider the series
$x_n=\sin(2\pi n/10)+\sin(2\pi n/15)$, $n=1,\ldots, N$, $N=150$, $L=70$,
which is depicted in Fig.~\ref{fig:fssa1}.

\begin{figure}[!htb]
    \centering
    \includegraphics{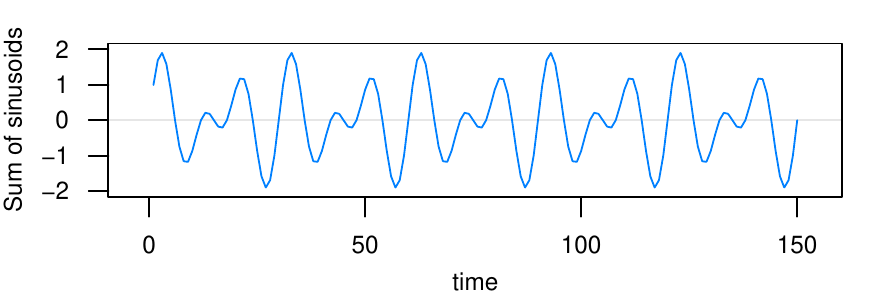}
    \caption{Sum of two sinusoids with equal amplitudes}
    \label{fig:fssa1}
\end{figure}

Sinusoids with periods 10 and 15 are approximately separable for such series and window lengths.
However, since the sinusoid amplitudes are equal, there is no strong separability and therefore
after Basic SSA we obtain an unsatisfactory decomposition, an arbitrary mixture of the sinusoids
(top picture of Fig.~\ref{fig:fssa2}) with $\bfw$-correlation between reconstructed by ET1--2 and ET3--4 series equal to 0.92.

The decomposition performed by DerivSSA with $\gamma=10$ applied to the group ET1--4 with $J_1=\{1,2\}$ and $J_2=\{3,4\}$ (Algorithm~\ref{alg:deriv}) is depicted in the bottom graph of Fig.~\ref{fig:fssa2} and demonstrates the very accurate separability,
$\bfw$-correlation is equal to 0.01. The second measure, the mean proportion $0.5(\tau_{r_1}(\tX^{(1)})+\tau_{r_2}(\tX^{(2)})$, is diminished from
0.3266 to 0.0003.
For this example, the obtained decomposition practically does not depend on $\gamma$ for all $\gamma>2$.

\begin{figure}[!htb]
    \centering
    \includegraphics{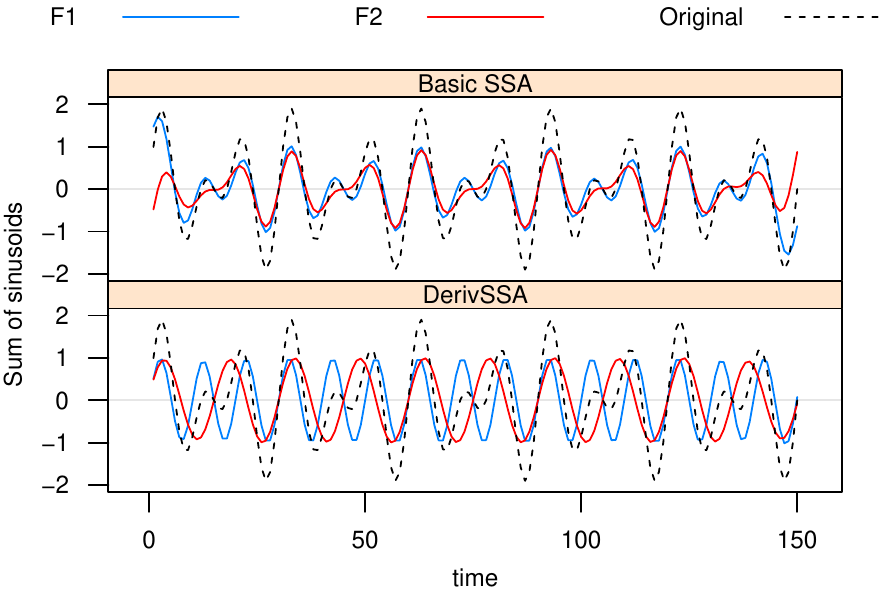}
    \caption{Sum of two sinusoids with equal amplitudes: reconstruction by Basic SSA (top) and DerivSSA (bottom)}
    \label{fig:fssa2}
\end{figure}

\section{Real-life time series}
\label{sec:examples}
In this section we apply Iterative O-SSA (Algorithm~\ref{alg:iter} and \ref{alg:LRcalc}
with possible modification provided by Algorithm~\ref{alg:sigmacor})
and DerivSSA (Algorithm~\ref{alg:deriv}) to real-life time series.
The role of the methods for separability of sine-waves was demonstrated in Sections
\ref{sec:close_freq} and \ref{sec:close_ampl} with the help of simulated data.
The obtained conclusions are generally valid for real-life series: DerivSSA adds to Basic SSA the ability to separate sine waves with close
amplitudes, while Iterative O-SSA can help in separation of sine waves, which are not orthogonal, that is, their
frequencies are insufficiently far one from another. Note that since in real-life series with seasonality there are no
close frequencies, DerivSSA can be very useful for seasonality decomposition.

In this section we consider the problem of trend extraction. The choice of examples is explained by the following considerations.

If a time series is long enough, then the oscillations are well weakly separated from the trend and only
strong separability is under question. Therefore, we expect that DerivSSA will work for trends of complex forms.

For short series, the trend can be not orthogonal to a periodic component like seasonality;
therefore, DerivSSA can even worsen the separability; moreover, derivation suppresses low-frequency components.
On the other hand, Iterative O-SSA is specially designed to separate non-orthogonal series components.

We will take only one iteration in Iterative O-SSA method, since it is sufficient to obtain good decomposition in the
considered examples and also makes the methods comparable by computational cost.

\subsection{Improving of strong separability}
Let us consider US Unemployment data (monthly, 1948-1981, thousands)
for male (20 years and over). Data are taken from \cite{Andrews.Herzberg1985},
the series length $N$ is equal to 408, see Fig.\ref{fig:usun_series}.
Since the series is long, we can expect weak separability of the trend and the seasonality.
For better weak separability we choose the window length equal to $L=N/2=204$, which is divisible by 12.

\begin{figure}[!htb]
    \centering
    \includegraphics{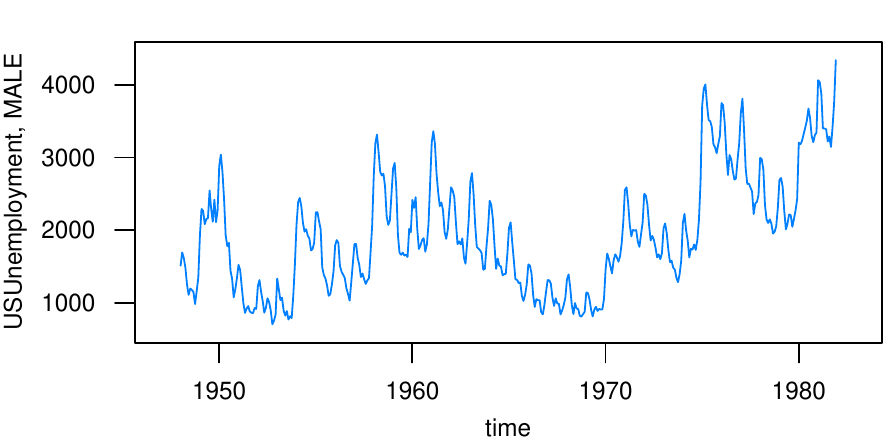}
    \caption{US unemployment: initial series}
    \label{fig:usun_series}
\end{figure}

\begin{figure}[!htb]
    \centering
    \includegraphics{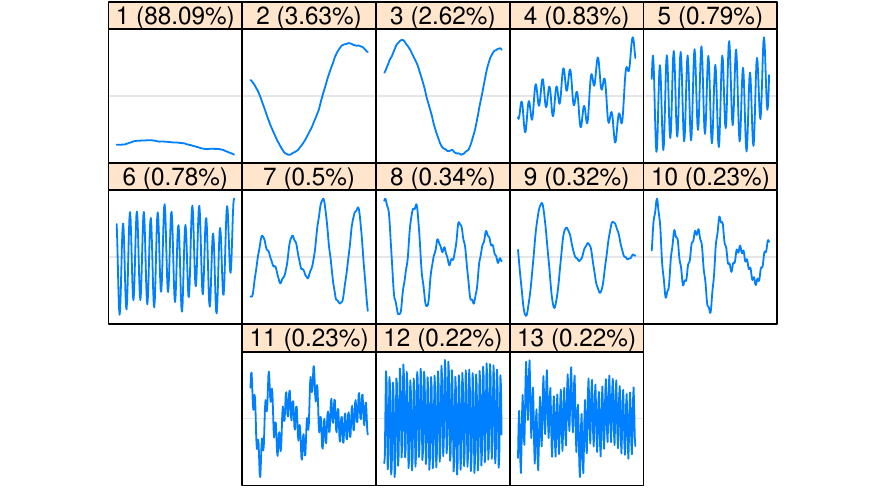}
    \caption{US unemployment: eigenvectors obtained by Basic SSA}
    \label{fig:usun_evect}
\end{figure}
\begin{figure}[!htb]
    \centering
    \includegraphics{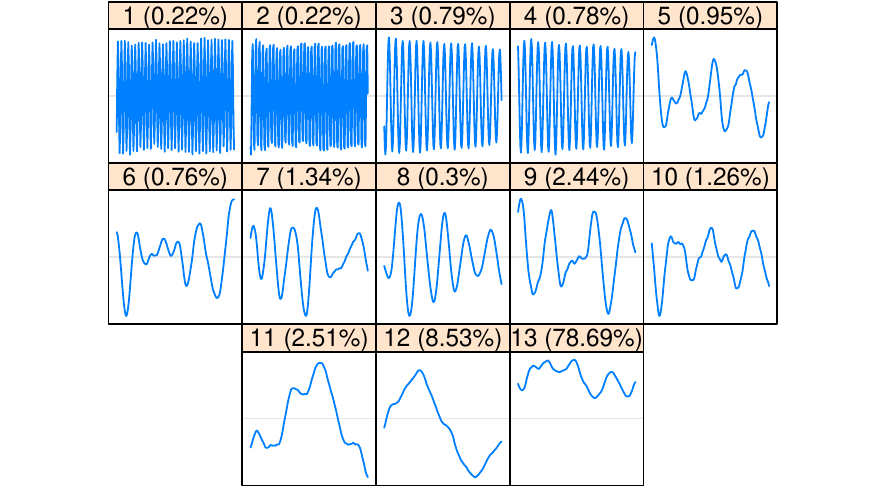}
    \caption{US unemployment: eigenvectors obtained by DerivSSA}
    \label{fig:usun_evect_fssa}
\end{figure}

Basic SSA does not separate the trend and seasonality (see Fig.~\ref{fig:usun_evect} and Fig.~\ref{fig:assa_empl1b} (left))
for this time series, likely due to lack of strong separability.
This is the typical situation when the trend has a complex form, trend components are mixed with the seasonality components
 and therefore the so called Sequential SSA was recommended \cite[Section 1.7.3]{Golyandina.etal2001}.
However, this is the case when DerivSSA should help.

We apply DerivSSA to the group ET1--13 that can be related to the signal.  DerivSSA separates different frequencies so
that components with higher frequencies become leading ones.
Since the low-frequency components in the considered series have large contribution, the weight of derivatives should be
large to make the seasonal components leading; we take $\gamma=1000$.

The resulting eigenvectors are depicted in Fig.~\ref{fig:usun_evect_fssa}.
One can see that the first 4 components contain seasonality, while the eigenvectors 5--13 contains components
of the trend. The mixture of the components within the trend group is not important.
Fig.~\ref{fig:usun_evect_fssa} demonstrates that the seasonal components are now separated from the residual.
Fig.~\ref{fig:assa_empl} depicting the DerivSSA reconstructions of the trend and the seasonality confirms that DerivSSA visibly improves the reconstruction
accuracy, especially at the ends of the series.

\begin{figure*}[!htb]
    \centering
    \includegraphics{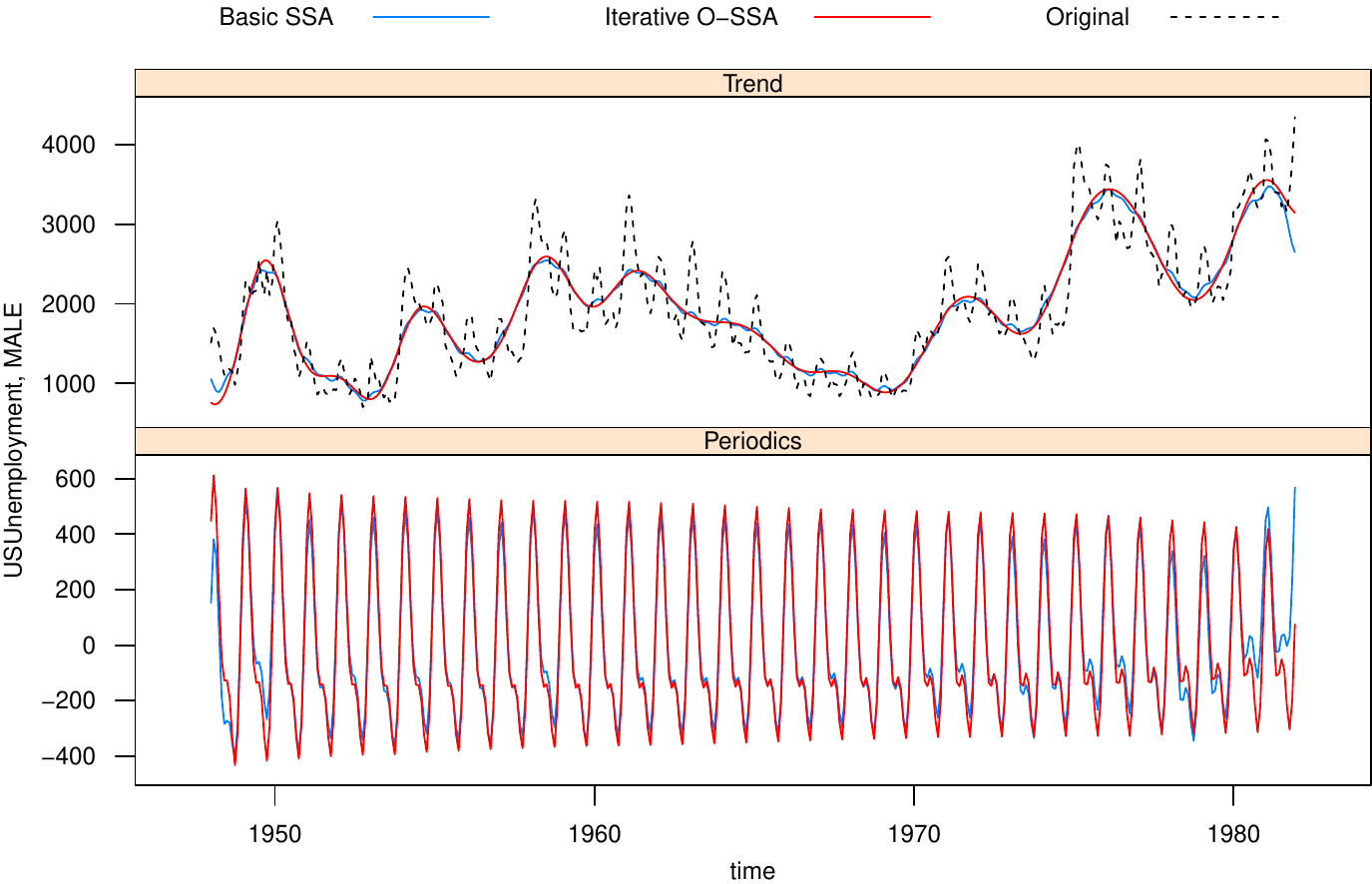}
    \caption{US unemployment: Decompositions by Basic SSA and Iterative O-SSA, which coincides with that by DerivSSA.}
    \label{fig:assa_empl}
\end{figure*}

Since Iterative O-SSA has possibility of sigma-correction, it also can help to move apart the decomposition components, and therefore we can
apply Iterative O-SSA to the group ET1--13 with the refined groups ET1--4,7--11 (trend) and ET5,6,12,13 (seasonality).
Since the components of the Basic SSA decomposition are mixed, we refer the components that contain mostly trend and slow cycles
to the first group and the components that contain mostly seasonality to the second group. As eigenvectors reflect
forms of the corresponding time series components, we can use the graph of eigenvectors shown in Fig.~\ref{fig:usun_evect}
for the initial grouping. For example,
the forth eigenvector looks like slow oscillations corrupted by seasonality and therefore we refer it to the trend group,
while the fifth eigenvector looks like seasonal component corrupted by something slow varying and we refer it to the seasonality group.
We apply one iteration with sigma-correction, taking $\varkappa=2$. After reordering caused by the sigma-correction,
the first trend group consists of the first eight components 1--8, while the second
seasonality group consists of 9--13 components, see Fig.~\ref{fig:usun_evect_assa}.

\begin{figure}[!htb]
    \centering
    \includegraphics{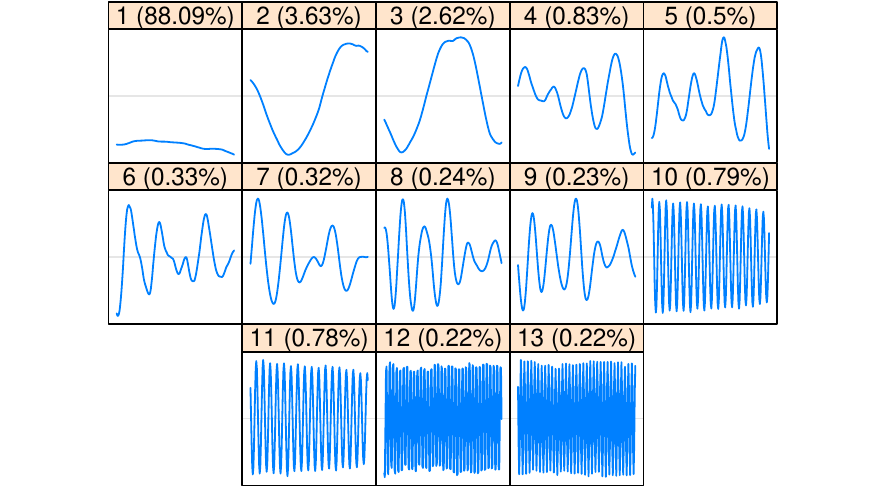}
    \caption{US unemployment: Iterative O-SSA eigenvectors}
    \label{fig:usun_evect_assa}
\end{figure}
\begin{figure}[!htb]
    \centering
    \begin{subfigure}{0.48\linewidth}
        \centering
        \includegraphics{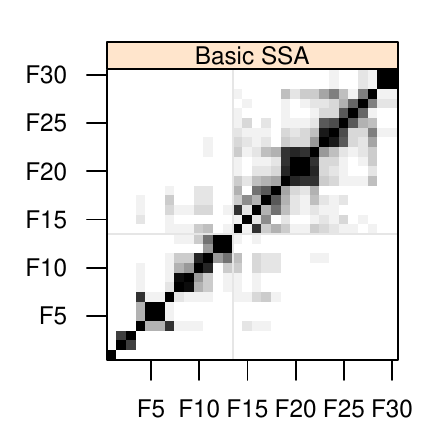}
    \end{subfigure}
    \begin{subfigure}{0.48\linewidth}
        \centering
        \includegraphics{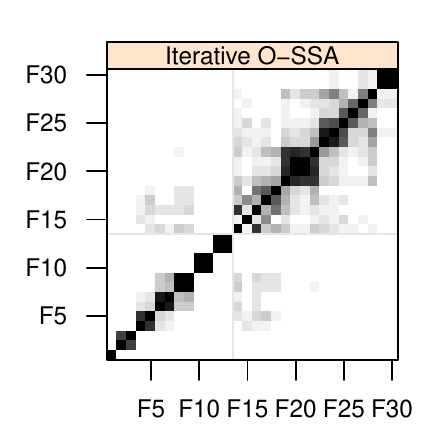}
    \end{subfigure}
    \centering
    \caption{US unemployment: $\bfw$-correlations before (left) and after (right) Iterative O-SSA}
    \label{fig:assa_empl1b}
\end{figure}
\begin{figure}[!htb]
    \centering
    \begin{subfigure}{0.48\linewidth}
        \centering
        \includegraphics{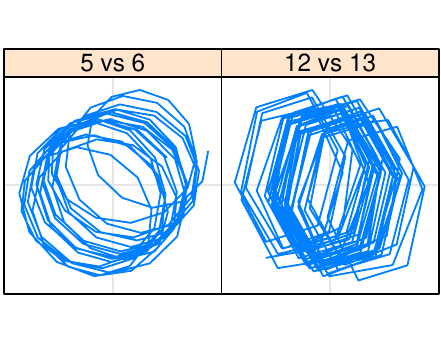}
    \end{subfigure}
    \begin{subfigure}{0.48\linewidth}
        \centering
        \includegraphics{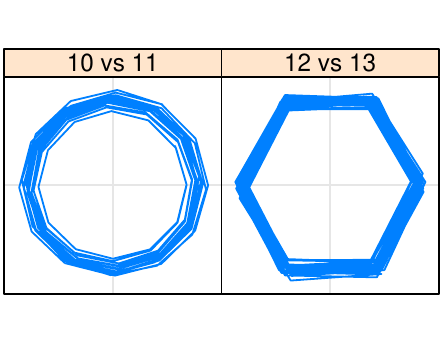}
    \end{subfigure}
    \caption{US unemployment: 2D plots of periodic eigenvectors before (left) and after (right) Iterative O-SSA}
    \label{fig:usun_paired_compare}
\end{figure}

The trend eigenvectors of the DerivSSA decomposition (Fig.~\ref{fig:usun_evect_fssa}, ET5--13) differ from that
of the O-SSA decomposition (Fig.~\ref{fig:usun_evect_assa}, ET1--8), the seasonality components are almost the same.
Nevertheless, the result of Iterative O-SSA reconstruction is visibly the same as that of DerivSSA shown in Fig.~\ref{fig:assa_empl}
and therefore we do not depict this reconstruction.

Fig.~\ref{fig:assa_empl1b} contains the $\bfw$-correlations between the elementary components provided by Basic SSA (left)
and the $\bfw$-correlations between the elementary components
reconstructed by Iterative O-SSA (right). The figure confirms the improving of separability.
Note that although an oblique decomposition was formally obtained, this decomposition is almost $\rmF$-orthogonal
(the maximal $\rmF$-correlation  between elementary matrix components, which is calculated as $\langle\bfX_i,\bfX_j\rangle_\rmF/(\|\bfX_i\|_\rmF\|\bfX_j\|_\rmF)$, is equal to 0.00368);
therefore, conventional $\bfw$-correlations are appropriate, see Appendix~\ref{sec:app}.
For trend extraction, it is important that correlations between trend and seasonality
groups are close to zero. Really, correlations between ET1--8 and ET9--13 are small.
Mixture of the components within the trend group is not important.
One can see that the trend components are still slightly mixed with the noise components.
However, we had a mixture with the residual before iterations (left) and
this cannot be corrected by Iterative O-SSA (right),
since the nested version is used.
Fig.~\ref{fig:usun_paired_compare} shows the improvement of separability with the help of
scatterplots of seasonal eigenvectors. After one iteration, plots of seasonal eigenvectors
form almost regular polygons.

Figures for the decomposition of DerivSSA analogous to Fig.~\ref{fig:assa_empl1b} and \ref{fig:usun_paired_compare} are very similar
and we do not present them in the paper.
Note that in DerivSSA we group components after their separation, what is easier
than to group mixing components for Iterative O-SSA before separation.
That is, in the considered example the resultant decomposition is the same, but application of DerivSSA is easier.

\subsection{Improving of weak separability}
Let us consider the series `Fortified wine'
(fortified wine sales, Australia, monthly, from January 1980, thousands of litres)
taken from \cite{Hyndman2013}.
The first 120 points of the series are depicted in Fig.~\ref{fig:fort2_rec}.

The series length is long enough to obtain weak separability; therefore, we will
consider short subseries to demonstrate the advantage of Iterative O-SSA for improving
of weak separability.

We take here the window length $L=18$ to make the difference between methods clearly visible on the figures, although the relation between
accuracies of the considered methods is the same for other choices of window lengths. Let us consider two subseries,
from 30th to 78th points and from 36th to 84th points. The difference consists in behavior of
the seasonality at the ends of the subseries.

\begin{figure}[!htb]
    \centering
    \includegraphics{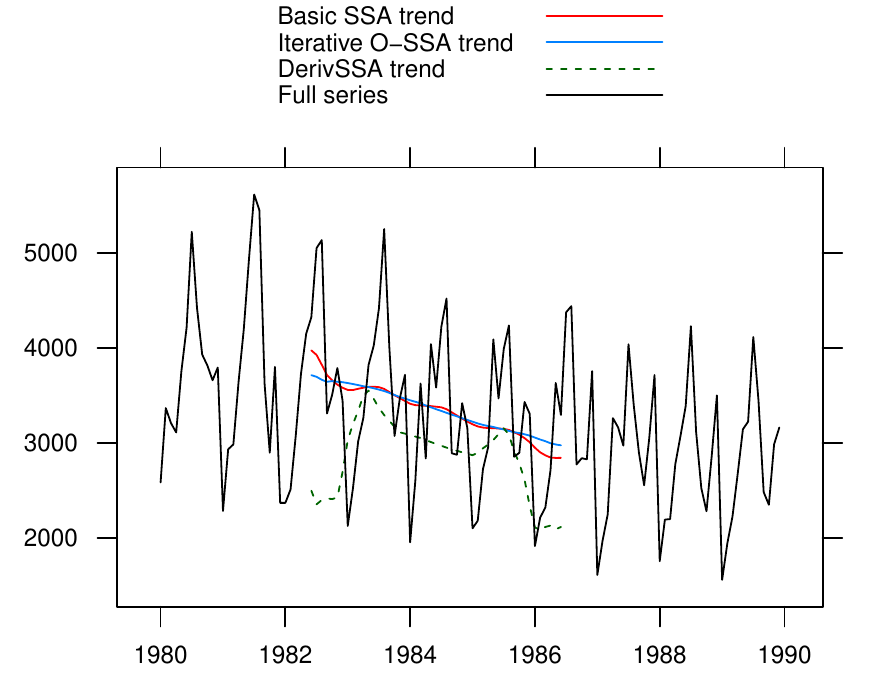}
    \caption{Fortified wines: trend reconstruction by DerivSSA and Iterative O-SSA for subseries of points 30--78. }
    \label{fig:fort2_rec}
\end{figure}

\begin{figure}[!htb]
    \centering
    \includegraphics{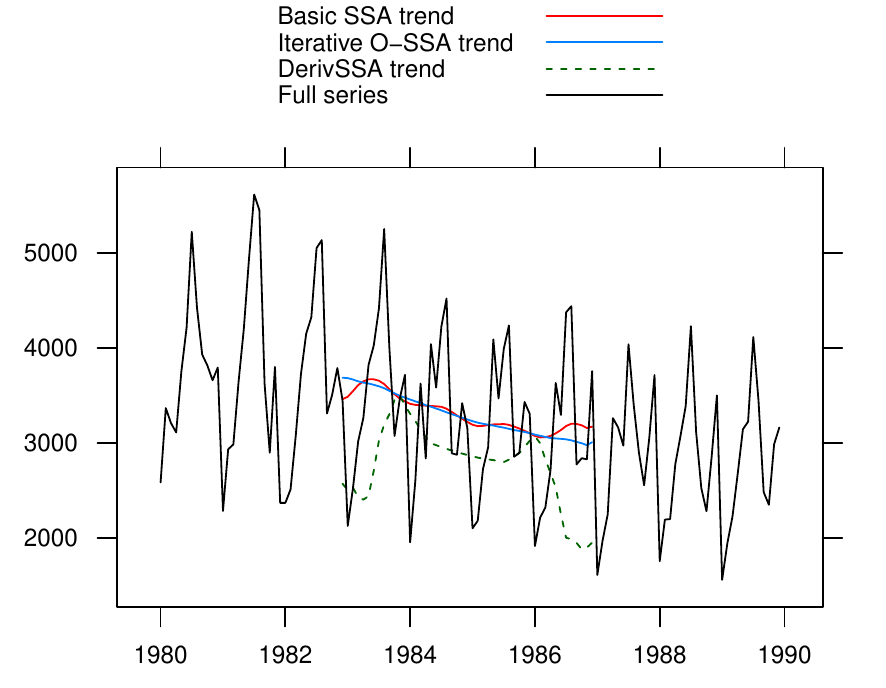}
    \caption{Fortified wines: trend reconstruction by DerivSSA and Iterative O-SSA for subseries of points 36--84. }
    \label{fig:fort1_rec}
\end{figure}

As well as in the previous example, we start with Basic SSA. ET1 is identified as corresponding
to trend, other components are produced by seasonality and noise (we do not include their pictures).
One can see in Fig.~\ref{fig:fort2_rec} and \ref{fig:fort1_rec} (red line) that the reconstructed trend
is slightly mixed with the seasonality and steps after the seasonality at the ends of the series.

To apply Iterative O-SSA, we should choose a group of elementary components
containing the trend components and approximately separated from the residual.
Let it be ET1--7. Thus, we apply one iteration of O-SSA to the refined groups ET1 and ET2--7.
Since the trend has the contribution much larger than the residual,
we consider Iterative O-SSA with no sigma-correction.
The result of reconstruction is much more relevant, see Fig.~\ref{fig:fort2_rec} and \ref{fig:fort1_rec} (blue line).
Green line in the same figures shows that DerivSSA gives more poor reconstruction than Basic SSA in this example.

\section{Conclusion}
\label{sec:conclusion}
We suggested two modifications of SSA, which can considerably improve the separability and thereby the reconstruction accuracy.
Iterative O-SSA shows its advantage dealing with separation of sine waves with close frequencies and with extraction
of trend for short series.
DerivSSA shows its advantage in conditions of weak separability dealing with long enough series
with complex-form trends and sine waves with equal amplitudes.

We demonstrated that for separation of trend even one iteration of Iterative O-SSA can improve the separability.
while DerivSSA works only in conditions of approximate weak separability.
On the other hand, for separability of weakly separable sine waves with equal amplitudes DeriveSSA works
more effectively than Iterative O-SSA.

The important aspect of both methods is that they should be applied to the estimated signal subspace (more general,
to the estimated subspace of the sum of components that we want to separate), that is,
they work in a nested manner. We can consider the methods as refining of the decomposition obtained by Basic SSA
(generally, the subspace estimation can be performed by any method, not necessarily by SSA).
Despite the both methods have the underlying model of series governed by linear recurrence relations,
the methods do not use the model directly. This allows one to apply the methods even if the signal
satisfies the model only locally. For example, the trend usually does not satisfy an LRR exactly; however, it
can be extracted by SSA and its considered variations.

The common part of the methods is the generalized SVD (so called Restricted SVD), which provides decompositions that are not
bi-orthogonal with respect to the conventional inner product. These methods do not use the optimality
properties of the generalized SVD; however, this is not essential for their success in the signal
decomposition.

The further development of the considered methods can consists in their combination for effective solution of the problem of
lack of both weak and strong separability and in the use of the obtained improved non-orthogonal decomposition for forecasting.

\appendix

\section{Inner products and related matrix decompositions}
\label{sec:app}

Here we provide the necessary information about matrix decompositions with respect to
given inner products in the row and column spaces (see e.g. \cite[Th.3]{VanLoan1976}), which are
called in \cite{DeMoor.Golub1991} Restricted SVD (RSVD).

\subsection{Inner products}
Usually, orthogonality of vectors in $\spaceR^M$ is considered in a conventional manner:
$X_1$ and $X_2$ in $\spaceR^M$ are orthogonal if their Euclidean inner product is equal to 0, i.e.
$(X_1,X_2)_M=0$, where $(\cdot,\cdot)_M$ is the standard inner product in $\spaceR^M$. Sometimes we will omit the dimension
in denotation if it is clear from the context.
It is well-known that any inner product in $\spaceR^M$ can be defined as
$\langle X_1,X_2 \rangle_\bfA=X_1^\rmT \bfA X_2$ for a symmetric positive-definite matrix $\bfA$.
For any $\bfO_\bfA$ such that $\bfO_\bfA^\rmT \bfO_\bfA=\bfA$ we have $\langle X_1,X_2 \rangle_\bfA=(\bfO_\bfA X_1,\bfO_\bfA X_2)_M$.
Evidently, $\bfO_\bfA$ is defined up to multiplication by an orthogonal matrix.

The inner product yields the notion of orthogonality.
We will say that two vectors are \emph{$\bfA$-orthogonal} if $\langle X_1,X_2 \rangle_\bfA=0$.

Let the matrix $\bfA$ be symmetric positive semi-definite, $\rank(\bfA)=r$.
Then $\bfA$ can be decomposed as $\bfA=\bfO_\bfA^\rmT \bfO_\bfA$ with $\bfO_\bfA \in \cM_{r,M}$
and $\langle X_1,X_2 \rangle_\bfA=(\bfO_\bfA X_1,\bfO_\bfA X_2)_r$.
Note that the row space of $\bfO_\bfA$ is the same for any choice of $\bfO_\bfA$ and coincides with the column
space of $\bfA$.
If the matrix $\bfA$ is not positive definite, then we obtain a degenerate inner product,
that is, if $\langle X,X \rangle_\bfA=0$, then it is not necessary that $X=0_M$.
However, for vectors belonging to the column space of $\bfA$ the equality $\langle X,X \rangle_\bfA=0$ yields $X=0_M$.
Thus, if we consider inner product generated by a rank-deficient matrix $\bfA$,
then we should consider it only on the column space of $\bfA$.
In particular, we can correctly define $\bfA$-orthogonality of vectors from the column space
of $\bfA$.

The following evident proposition shows that any basis can be considered as $\bfA$-orthonormal for some
choice of $\bfO_\bfA$.

\begin{proposition}
\label{prop:orth}
  Let $P_1,\ldots,P_r$ be a set of linearly independent vectors in $\spaceR^M$.
  Then $P_1,\ldots,P_r$ are $\bfA$-orthonormal for
  $\bfO_\bfA=\bfP^\dag$, where $\bfP=[P_1:\ldots:P_r]$.
\end{proposition}

Note that the column space of $\bfP$ coincides with the row space of $\bfO_\bfA$.
We call a matrix $\bfO_\bfA$ that makes a set $P_1,\ldots,P_r$ $\bfA$-orthonormal
\emph{orthonormalizing matrix} of this set. Certainly, the orthonormalizing matrix
is not uniquely defined.

\subsection{Oblique decompositions}
Let us consider a minimal decompositions of $\bfY\in \cM_{L,K}$ of rank $r$ in the form
\begin{equation}
\label{eq:LRSVD}
  \bfY=\sum_{i=1}^r \sigma_i P_i Q_i^\rmT,
\end{equation}
where $\sigma_1 \geq \sigma_2 \geq \ldots\geq \sigma_r>0$, $\{P_i\}_{i=1}^r$ and $\{Q_i\}_{i=1}^r$ are linearly independent (therefore,
$\{P_i\}_{i=1}^r$ is a basis of the column space of $\bfY$, $\{Q_i\}_{i=1}^r$ is a basis of the row space of $\bfY$).
It is convenient to write \eqref{eq:LRSVD} in the matrix form: $\bfY=\bfP\mathbf{\Sigma}\bfQ^{\rmT}$, where
$\bfP=[P_1:\ldots:P_r]$, $\bfQ=[Q_1:\ldots:Q_r]$ and $\mathbf{\Sigma}=\diag(\sigma_1,\ldots,\sigma_r)$.

\begin{proposition}
\label{prop:babilon}
Let $\bfO_\bfL$ be an orthonormalizing matrix of $\{P_i\}_{i=1}^r$ and $\bfO_\bfR$ be an orthonormalizing matrix of $\{Q_i\}_{i=1}^r$. Then
\begin{equation}
\label{eq:defLRSVD}
  \bfO_\bfL\bfY\bfO_\bfR^\rmT=\sum_{i=1}^r \sigma_i (\bfO_\bfL P_i) (\bfO_\bfR Q_i)^\rmT
\end{equation}
is an SVD of $\bfO_\bfL\bfY\bfO_\bfR^\rmT\in \cM_{r,r}$ with the left singular vectors $\bfO_\bfL P_i\in \spaceR^r$ and
the right singular vectors $\bfO_\bfR Q_i\in \spaceR^r$.
\end{proposition}

This proposition follows from the fact that any bi-orthogonal decomposition is an SVD.

\begin{definition}
\label{def:agreed}
If the column space of $\bfL$ contains the column space of $\bfY$ and the column space of $\bfR$ contains the row space of $\bfY$,
then we will call such a pair $(\bfL,\bfR)$ \emph{consistent with} the matrix $\bfY$.
\end{definition}

\begin{definition}
\label{def:LRSVD}
For $(\bfL,\bfR)$ consistent with $\bfY$, we say that \eqref{eq:LRSVD} is an $(\bfL,\bfR)$-SVD, if the system $\{P_i\}_{i=1}^r$ is $\bfL$-orthonormal and
the system $\{Q_i\}_{i=1}^r$ is $\bfR$-orthonormal.
\end{definition}

In a matrix statement of problem \cite{DeMoor.Golub1991}, the $(\bfL,\bfR)$-SVD is called Restricted SVD of
$\bfY$ with respect to $(\bfL,\bfR)$.

It follows from Definition~\ref{def:LRSVD} that \eqref{eq:defLRSVD} is an SVD  if and only if \eqref{eq:LRSVD} is an $(\bfL,\bfR)$-SVD,
where $\bfL=\bfO_\bfL^\rmT \bfO_\bfL$ and $\bfR=\bfO_\bfR^\rmT \bfO_\bfR$, $\bfO_\bfL$ and $\bfO_\bfR$ are orthonormalizing.

Proposition~\ref{prop:babilon} says that any minimal decomposition into a sum of matrices of rank 1 in the form
\eqref{eq:LRSVD} is the $(\bfL,\bfR)$-SVD for some matrices $\bfL$ and $\bfR$.
%




\begin{proposition}
\label{prop:calcLRSVD}
Let
\begin{equation}
\label{eq:LVSVDcalc}
  \bfO_\bfL\bfY\bfO_\bfR^\rmT=\sum_{i=1}^r \sqrt{\lambda_i} U_i V_i^\rmT
\end{equation}
be the ordinary SVD of $\bfO_\bfL\bfY\bfO_\bfR^\rmT$.
Then the decomposition \eqref{eq:LRSVD} with $\sigma_i=\sqrt{\lambda_i}$, $P_i=\bfO_\bfL^\dag U_i$ and $Q_i=\bfO_\bfR^\dag V_i$ is the $(\bfL,\bfR)$-SVD.
\end{proposition}

Proposition~\ref{prop:calcLRSVD} follows from Proposition~\ref{prop:babilon} and provides the method how the $(\bfL,\bfR)$-SVD can be calculated (see Algorithm~\ref{alg:LRSVD}).

Let us show how we can change the set of $\sigma_i$ in the $(\bfL,\bfR)$-SVD \eqref{eq:LRSVD} without change
of directions of $P_i$ and $Q_i$, that is, of $P_i/\|P_i\|$ and $Q_i/\|Q_i\|$.

\begin{proposition}
\label{prop:reord}
Let \eqref{eq:LRSVD} be the $(\bfL,\bfR)$-SVD with $\bfO_\bfL=\bfP^\dag$ and $\bfO_\bfR=\bfQ^\dag$.
Then
\begin{equation}
\label{eq:LRSVDsigma}
  \bfY=\sum_{i=1}^r \wtilde\sigma_i \wtilde{P}_i \wtilde{Q}_i^\rmT,
\end{equation}
where $\wtilde\sigma_i=\sigma_i/(\mu_i\nu_i)$, $\wtilde{P}_i=\mu_i P_i$ and $\wtilde{Q}_i=\nu_i Q_i$,
is (after reordering of $\wtilde\sigma_i$) the $(\wtilde\bfL,\wtilde\bfR)$-SVD with
$\bfO_{\wtilde\bfL}=\wtilde\bfP^\dag$ and $\bfO_{\wtilde\bfR}=\wtilde\bfQ^\dag$.
\end{proposition}

The case of one-side non-orthogonal decompositions, when one of the matrices, $\bfR$ or $\bfL$, is identical, is of special concern.
It is shown in \cite{DeMoor.Golub1991} that then Restricted SVD is Quotient SVD (often called Generalized SVD \cite{Paige.Saunders1981}).
If $\bfL$ is the identity matrix, then $P_i$, $i=1,\ldots,r$, are orthonormal in the conventional sense and
form an orthonormal basis of the column space of $\bfY$.
If $\bfR$ is the identity matrix, then $Q_i$, $i=1,\ldots,r$, are orthonormal and
constitute an orthonormal basis of the row space.

\subsection{Matrix scalar products and approximations}
\label{sec: minnerprod}

Let $\bfX, \bfY \in \cM_{L,K}$, $(\bfL,\bfR)$ be consistent with both $\bfX$ and $\bfY$, $\bfL=\bfO_\bfL^\rmT \bfO_\bfL$
and $\bfR=\bfO_\bfR^\rmT \bfO_\bfR$.

Define the induced Frobenius inner product as
\bea
\langle \bfX, \bfY \rangle_{\rmF,(\bfL,\bfR)}=\langle\bfO_\bfL\bfX\bfO_\bfR^\rmT, \bfO_\bfL\bfY\bfO_\bfR^\rmT\rangle_\rmF.
\eea
Note that the definition does not depend on the choice of $\bfO_\bfL$ and $\bfO_\bfR$,
since $\langle\bfC,\bfD\rangle_\rmF=\tr(\bfC^\rmT\bfD)=\tr(\bfC\bfD^\rmT)$.

For two matrices $\bfC$ and $\bfD$ we say that they \\
1. $(\bfL,\bfR)$ $\rmF$-orthogonal if
$\langle \bfC, \bfD \rangle_{\rmF,(\bfL,\bfR)}=0$, \\
2. $(\bfL,\bfR)$-left orthogonal if $\bfC^\rmT\bfL\bfD=\bf0$, \\
3. $(\bfL,\bfR)$-right orthogonal if $\bfC\bfR\bfD^\rmT=\bf0$,\\
4. $(\bfL,\bfR)$ bi-orthogonal if the left and right orthogonalities hold.

Left or right orthogonality is the sufficient condition for $\rmF$-orthogonality.
The matrix components of an $(\bfL,\bfR)$-SVD are $(\bfL,\bfR)$ bi-orthogonal and therefore $(\bfL,\bfR)$ $\rmF$-orthogonal.

The measure of $(\bfL,\bfR)$ orthogonality is
\begin{equation}
\label{eq:ABcorr}
\rho_{(\bfL,\bfR)}(\bfX, \bfY)=\frac{\langle \bfX, \bfY \rangle_{\rmF,(\bfL,\bfR)}}
{\| \bfX \|_{\rmF,(\bfL,\bfR)} \| \bfY \|_{\rmF,(\bfL,\bfR)}}.
\end{equation}

Let $\bfX=\sum_i \bfX_i$, where $\bfX_i=\sigma_i P_i Q_i^\rmT$, be the $(\bfL,\bfR)$-SVD. Then $\| \bfX_i \|_{\rmF,(\bfL,\bfR)}=\sigma_i$ and $\| \bfX \|^2_{\rmF,(\bfL,\bfR)}=\sum_i \sigma_i^2$.
The contribution of $\bfX_k$ is equal to $\sigma_k^2 / \sum_i \sigma_i^2$.

The following proposition follows from the representation of the Frobenius scalar product
through the trace of matrix multiplication.
\begin{proposition}
  If $\bfX$ and $\bfY$ are $(\bfL,\bfR)$ left-orthogonal, then $\bfX$ and $\bfY$ are $(\bfL,\wtilde\bfR)$ $\rmF$-orthogonal for any $\wtilde\bfR$.
\end{proposition}

\begin{corollary}
\label{col:F_orth}
  Let $\bfL$ be the identity matrix and $\bfX$ and $\bfY$ be $(\bfL,\bfR)$ left-orthogonal for some matrix $\bfR$.
  Then  the conventional $\rmF$-orthogonality of $\bfX$ and $\bfY$ holds and $\|\bfX+\bfY\|_\rmF^2=\|\bfX\|_\rmF^2+\|\bfY\|_\rmF^2$.
\end{corollary}

Corollary~\ref{col:F_orth} shows that if at least in either row or column matrix spaces the conventional inner product is given,
that is, vectors are orthogonal in the ordinary sense,
then the conventional $\rmF$-orthogonality can be considered and
$\rmF$-norm and $\rmF$-inner product can be used to measure the approximation accuracy and the component orthogonality.

\begin{remark}
\label{rem:noagree}
The introduced definitions and statements are appropriate
if $\bfL$ and $\bfR$ are consistent with the matrices $\bfX$ and $\bfY$ (see Definition~\ref{def:agreed}).
Otherwise, e.g., \eqref{eq:ABcorr} can be formally calculated, but this measure will reflect only the correlation between
projections of columns and rows of $\bfX$ and $\bfY$ on the row spaces of $\bfL$ and $\bfR$ correspondingly.
\end{remark}

Let us remark that the conventional Frobenius norm is an interpretable characteristic of approximation, while
the norm based on  $\langle \cdot, \cdot \rangle_{\rmF,(\bfL,\bfR)}$ is much worse interpretable,
since it is equivalent to approximation by the Frobenius norm of the matrix
$\bfO_\bfL \bfX \bfO_\bfR^\rmT$.

\section*{Acknowledgments}
The authors are grateful to the anonymous reviewers and the editor for their useful comments and suggestions,
and to Konstantin Usevich for fruitful discussions which helped to considerably improve the paper.


\end{document}